\documentclass[journal,twoside,romanappendices]{IEEEtran}
\ifCLASSINFOpdf
\else
   \usepackage[dvips]{graphicx}
\fi

\usepackage{cite}
\usepackage{amsfonts}
\usepackage{lpic}
\usepackage{psfrag}
\usepackage{epsf} 
\usepackage{dsfont}
\usepackage{textcomp}
\usepackage{stmaryrd}
\usepackage{amssymb}
\usepackage{mathrsfs}
\usepackage{textcomp}
\usepackage{lettrine} 
\usepackage{flushend}
\usepackage{amsbsy}
\usepackage[cmex10]{amsmath}
\usepackage{accents}
\usepackage{array}
\usepackage{mdwmath}
\usepackage{mdwtab}
\usepackage{amsthm}
\usepackage{color}
\usepackage{graphicx}
\usepackage{caption}
\usepackage{subcaption}
\usepackage{multirow}
\usepackage{empheq}
\usepackage{color}
\definecolor{links}{rgb}{0.7,0,0}   
\definecolor{urls}{rgb}{0,0,0.8}    
\definecolor{cites}{rgb}{0,0,0.8}   

\usepackage[colorlinks,hyperindex,linkcolor=links,citecolor=cites,urlcolor=urls]{hyperref}

\usepackage{accents}
\newcommand{\ubar}[1]{\underaccent{\bar}{#1}}

\newcommand{\sigd}{\sigma^2_{\Delta}}
\newcommand{\sigN}{\sigma^2_\text{W}}
\newcommand{\rt}{{\bf \tilde{r}}}
\newcommand{\Rt}{{\bf \tilde{R}}}
\newcommand{\bx}{{\bf x}}
\newcommand{\by}{{\bf y}}
\newcommand{\br}{{\bf r}}
\newcommand{\bR}{{\bf R}}
\newcommand{\bt}{{\boldsymbol \theta}}
\newcommand{\bT}{{\boldsymbol \Theta}}

\newcommand{\Es}{{E_\text{s}}}
\newcommand{\SNR}{\mathsf{SNR}}
\newcommand{\Ex}[2]{{\mathbb{E}_{#1}\left[#2\right]}} 	
\newcommand{\alphaU}{\alpha_\text{U}}
\newcommand{\betaU}{\beta_\text{U}}
\newcommand{\alphaL}{\alpha_\text{L}}
\newcommand{\betaL}{\beta_\text{L}}
\newtheoremstyle{mystyle}
  {}
  {}
  {}
  {}
  {\itshape}
  {:}
  { }
  {\thmname \thmnumber{{#1}{ }#2}}

\theoremstyle{mystyle}
\newtheorem{theorem}{Theorem}
\newtheorem{prop}[theorem]{Proposition}
\newtheorem{lem}[theorem]{Lemma}

\begin{document}
\title{\huge{On the Capacity of the Wiener Phase-Noise Channel:\\ Bounds and Capacity Achieving Distributions}}
\author{\hspace{0in}\IEEEauthorblockN{M. Reza Khanzadi, \emph{Student Member, IEEE}, Rajet Krishnan, \emph{Student Member, IEEE}, \\Johan S\"{o}der, Thomas Eriksson}

\thanks{M. Reza Khanzadi is with the Department of 
Signals and Systems, and also Department of Microtechnology and Nanoscience, Chalmers University of Technology, 41296 Gothenburg, Sweden. (email: khanzadi@chalmers.se.)

Rajet Krishnan, and Thomas Eriksson are with the Department of 
Signals and Systems, Chalmers University of Technology, 41296 Gothenburg, Sweden. (email: {rajet,thomase}@chalmers.se.)

Johan S\"{o}der is with Ericsson Research, Stockholm, Sweden. (email: johan.soder@ericsson.com)

The simulations were performed on resources at Chalmers Centre for Computational Science and Engineering (C3SE) provided by the Swedish National Infrastructure for Computing (SNIC). 
}
}


\maketitle

\begin{abstract}
In this paper, the capacity of the additive white Gaussian noise (AWGN) channel, affected by time-varying Wiener phase noise is investigated. Tight upper and lower bounds on the capacity of this channel are developed. The upper bound is obtained by using the duality approach, and considering a specific distribution over the output of the channel. In order to lower-bound the capacity, first a family of capacity-achieving input distributions is found by solving a functional optimization of the channel mutual information. Then, lower bounds on the capacity are obtained by drawing samples from the proposed distributions through Monte-Carlo simulations. The proposed capacity-achieving input distributions are circularly symmetric, non-Gaussian, and the input amplitudes are correlated over time. The evaluated capacity bounds are tight for a wide range of signal-to-noise-ratio (SNR) values, and thus they can be used to quantify the capacity. Specifically, the bounds follow the well-known AWGN capacity curve at low SNR, while at high SNR, they coincide with the high-SNR capacity result available in the literature for the phase-noise channel.

\end{abstract}


\section{Introduction}
\lettrine[lines=2]{P}HASE NOISE due to frequency instabilities of radio frequency oscillators is a limiting factor in high data rate digital communication systems (e.g., see \cite{Viterbi1963,dallal92,Tomba98,Colavolpe2005, syrjala2009phase, syrjala2011ofdm, Khanzadi2011, Mehrpouyan2012,Krishnan2012_1, Khanzadi2013_0_COMP,Khanzadi2013_1, syrjala2014analysis, Bjornson_MAMIMO2014_1,Krishnan14-01a,krishnan2015_VT_linear_MIMO} and references therein). Phase noise severely impacts the performance of systems that employ dense signal constellations \cite{Krishnan2013_1,Krishnan13_2}. Moreover, the effect of phase noise is more pronounced in high carrier frequency systems, e.g., E-band (60-80 GHz), mainly due to the high levels of phase noise in oscillators designed for such frequencies \cite{Smulders2002GHz,li2003high,Dohler2011,Mehrpouyan2014_EBAND,khanzadi2014_highFreq}. 

The Shannon capacity of the system can be studied in order to investigate the effect of phase noise on the throughput. For stationary phase-noise channels, Lapidoth \cite{lapidoth02-10a} derived an asymptotic capacity expression, that is valid at high SNR. Capacity of the noncoherent channel, where the transmitted signal is affected by uniformly distributed phase noise, has been studied in \cite{colavolpe2001capacity, peleg98-05a,nuriyev05-03a,katz04-10a}. In \cite{katz04-10a}, Katz and Shamai derived upper and lower bounds on the capacity of the noncoherent phase-noise channel for non-asymptotic SNR regimes. They showed that the capacity-achieving distribution of the noncoherent memoryless channel is discrete with an infinite number of mass points. In \cite{durisi12-08a,colavolpe2001capacity}, the capacity bounds in \cite{katz04-10a} have been extended to the block memoryless phase-noise channel, where the phase noise was modeled as a constant over a number of consecutive symbols. It was shown in \cite{colavolpe2001capacity} that the capacity-achieving input distribution of the block memoryless phase-noise channel is not Gaussian (unlike the additive white Gaussian noise channel). In \cite{peleg98-05a}, the constrained capacity of M-ary phase-shift keying over a noncoherent phase-noise channel has been investigated. Capacity of partially coherent channels, where the phase noise is estimated at the receiver, and the signal is affected by the residual phase noise estimation errors, has been studied in \cite{hou02-05a}. The achievable information
rate of phase-noise channel and methods for the calculation of that have been
discussed in, e.g.,~\cite{barletta11-11a,barbieri11-12a,barletta12-05a}. Achievable information rate of multi-carrier radio links, affected by phase noise, has been analyzed in~\cite{gokceoglu2013mutual}. Effects of using multisampling receivers on the achievable information rate of the phase-noise channel has been recently investigated in~\cite{ghozlan13-07a,ghozlan14-01a,barletta14-01a}.

There has been limited number of studies on characterizing the capacity of channels affected by phase noise with memory~(e.g.,~\cite{lapidoth02-10a,durisi2013capacity,durisi13-09b,ghozlan13-07a}). The Wiener phase-noise channel that models many practical scenarios belongs to this family of channels. In~\cite{lapidoth02-10a}, Lapidoth characterized the capacity of the Wiener phase-noise channel at high SNR. It was shown in~\cite{lapidoth02-10a} that circularly symmetric input alphabets with Gamma-distributed amplitudes can achieve the capacity of the stationary phase-noise channel (with or without memory) at high SNR. Capacity results of~\cite{lapidoth02-10a} have been recently extended to multi-antenna systems in~\cite{durisi2013multiplexing,durisi2013capacity,Khanzadi14-T01a}. However, the capacity-achieving input distribution of the Wiener phase-noise channel and the closed-form capacity of this channel, valid for all SNR values, have not been derived yet.  

\subsection*{Contributions}

In this paper, we derive tight upper and lower bounds on the capacity of the additive white Gaussian noise (AWGN) channel affected by Wiener phase noise when the channel input is subject to an average-power constraint. The upper bound on the capacity is found by using the duality approach, and considering a specific distribution over the output of the channel. We determine a family of input distributions that result in a tight lower bound on the capacity. We show that the capacity-achieving input distribution is circularly symmetric but non-Gaussian. We also show that unlike for memoryless channels, the capacity-achieving input alphabets are correlated over time. Lower bounds on the capacity are obtained by numerical calculation of the information rates, achievable by samples generated from the proposed input distributions through Monte-Carlo simulations. The developed upper and lower bounds are tight for a wide range of SNR values. This helps to more accurately quantify the capacity of the phase-noise channel compared to previously available results in the literature~(e.g.,~\cite{lapidoth02-10a}).

\subsection*{Organization of the Paper}
The paper is organized as follows. In Section~\ref{sec:System Model}, the system model and the corresponding amplitude-phase channel are introduced. Using the amplitude-phase channel, the mutual information between the input and the output is examined in Section~\ref{sec:Mutual Information}. In Section~\ref{sec:upper_bound}, a capacity upper bound is derived. In Section~\ref{sec:lower_bound}, we obtain the closed-form expression for a family of capacity-achieving distributions. Finally, in Section~\ref{sec:Results}, the proposed lower and upper bounds are compared against each other and also the results available in the literature.

\subsection*{Notation}

With $\mathcal{N}(0,\sigma^2)$ and $\mathcal{CN}(0,\sigma^2)$, we denote the probability distribution of a real Gaussian random variable, and of a circularly symmetric complex Gaussian random variable with zero mean and variance $\sigma^2$. The uniform distribution over the interval $[0,2\pi)$ is denoted as~$\mathcal{U}(0,2\pi)$. We use~$|\cdot|$ to denote the absolute value of scalars, and determinant of matrices. The Euclidean norm of vectors is denoted by $||\cdot||$. Finally,~$\mathcal{D}(\cdot||\cdot)$ denotes the relative entropy between two probability distributions. For notational convenience, $f(x)$ and $f(y)$ refer to two different  probability distribution functions, $f_x(x)$ and $f_y(x)$, respectively.
\section{System Model}\label{sec:System Model}
\subsection{The channel}
The input-output relation of the discrete-time Wiener phase-noise channel can be written as \cite{Colavolpe2005}
\begin{align}
y_k = x_k e^{j\phi_k} + w_k, \label{eq:channel}
\end{align}
where $x_k$ is the transmitted symbol, and $w_{k}$ is circular symmetric AWGN independently distributed from $\mathcal{CN}(0,2\sigN)$.
The process, $\phi_k$, is the discrete-time Wiener phase noise 
\begin{align}
&\phi_k = \phi_{k-1} + \Delta_k, &&\Delta_k \sim \mathcal{N}(0,\sigma^2_{\Delta}).
\end{align}
This discrete-time process corresponds to a sampled version of a continuous-time Brownian motion process with uncorrelated increments.\footnote{For discussions on the limitations of the Wiener phase noise model see \cite{Khanzadi2011,Khanzadi2013_2_ColoredPNEst,Khanzadi2013_1}.}  Samples are taken every $T_\text{s}$ seconds, the transmission symbol interval.\footnote{Note that the system model~\eqref{eq:channel} is derived under the assumption that the continuous-time phase-noise process remains constant over the duration of the symbol time. This assumption allows us to obtain a discrete-time equivalent channel model by sampling at Nyquist rate. As shown recently in \cite{ghozlan13-07a,ghozlan14-01a,barletta14-01a}, by dropping this assumption one may obtain different high-SNR capacity characterization.}  The continuous time process of the corresponding oscillator has a Lorentzian spectrum \cite{Demir2000,Khanzadi2011}. This spectrum is fully characterized by a single parameter; the $3$dB single-sided bandwidth, $f_{3\text{dB}}$, which depends on central frequency and design technology of the oscillator~\cite{khanzadi2014_highFreq}. The innovation variance for the discrete-time phase-noise process is $\sigma^2_{\Delta} = 4\pi f_{3\text{dB}} T_\text{s}$.\footnote{Note that the innovation variance can equivalently be found directly from the spectrum of the phase-noise process~\cite{khanzadi2014_highFreq}.}

\subsection{Amplitude and phase input-output relations} \label{subsec:amplitude phase}
The input $x_k$, to the channel \eqref{eq:channel} and the output $y_k$ are complex numbers, and can be represented in polar form as $x_k = R_k e^{j\Theta_k}$, and $y_k = r_k e^{j\theta_k}$. In this notation, $R_k$ and $\Theta_k$ denote the amplitude and the phase of the transmitted symbol $x_k$, while $r_k$ and $\theta_k$ denote the amplitude and the phase of the received sample $y_k$, respectively. The input-output relations between the transmitted and received amplitude and phase are
\begin{align}
&r_k  = \sqrt{(R_k + w_{k,\parallel})^2 + w^2_{k,\perp}} \label{eq:r_k} \\
&\theta_k = \Theta_k +  N_k  + \phi_k,  \label{eq:theta_k}
\end{align}
where 
\begin{align}
N_k=\arctan \left(\frac{w_{k,\perp}}{R_k+w_{k,\parallel}}\right),
\end{align}
and $w_{k,\parallel}$ and $w_{k,\perp}$ denote the parts of $w_k$ that are parallel~(in-phase) and orthogonal to the transmitted signal, respectively. In the rest of the paper, we investigate the capacity of the phase-noise channel~\eqref{eq:channel} by considering the equivalent amplitude and phase channels stated in~\eqref{eq:r_k} and \eqref{eq:theta_k}.

\subsection{Definition of Capacity}
The capacity of the phase-noise channel \eqref{eq:channel} is given by \cite{cover06-a} 
\begin{align}
C (\SNR) &=  \lim_{n \to \infty} \underset{f(\bx)}{\sup} \; \: \frac{1}{n}I(\bx;\by)  \label{eq:capacity def}\\
&=\lim_{n \to \infty} \underset{f(\bR,\bT)}{\sup} \; \: \frac{1}{n} I(\br,\bt;\bR,\bT),  \label{eq:capacity_polar_def}
\end{align}
where~$\bx=\{x_k\}_{k=1}^n$,~$\by=\{y_k\}_{k=1}^n$,~$\br=\{r_k\}_{k=1}^n$,~$\bt=\{\theta_k\}_{k=1}^n$,~$\bR=\{R_k\}_{k=1}^n$, and~$\bT=\{\Theta_k\}_{k=1}^n$.  
The supremum in~\eqref{eq:capacity def} and~\eqref{eq:capacity_polar_def} is computed over all probability distributions on the input that satisfies 
\begin{align}\label{eq:power_constraint}
\frac{1}{n} \sum_{k=1}^n\Ex{}{|x_k|^2} = \frac{1}{n} \sum_{k=1}^n\Ex{}{R_k^2} \leq \Es,
\end{align}
where~$\Es$ is the maximum average power. 
The SNR is defined as $\SNR=\Es/2\sigN$ throughout the paper.

\section{Mutual Information}\label{sec:Mutual Information}
The mutual information on the right-hand side~(RHS) of~\eqref{eq:capacity_polar_def} is written as
\begin{align}
\frac{1}{n} I(\br,\bt;\bR,\bT) \\
&\hspace{-2cm}=\frac{1}{n}\left( h(\br,\bt) - h(\br,\bt|\bR,\bT) \right) \\
&\hspace{-2cm}=\frac{1}{n}\left( h(\br) + h(\bt|\br) - h(\br|\bR,\bT) - h(\bt|\br,\bR,\bT) \right)\label{eq:mutual chain} \\
&\hspace{-2cm}= \frac{1}{n}\left( h(\br) + h(\bt|\br)- h(\br|\bR) - h(\bt|\br,\bR,\bT) \right) \label{eq:inf channels}\\ 
&\hspace{-2cm}= \frac{1}{n}\left( I(\br;\bR) + h(\bt|\br) - h( \bT + {\bf N} + {\boldsymbol \phi} | \br,\bT,\bR ) \right) \label{eq:inf channels_2}\\
&\hspace{-2cm}= \frac{1}{n}\left( I(\br;\bR) + h(\bt|\br) - h({\bf N} + {\boldsymbol \phi} | \br,\bR ) \right), \label{eq:inf channels_3}
\end{align}
where~${\boldsymbol \phi}=\{\phi_k\}_{k=1}^n$, and~${\bf N}=\{N_k\}_{k=1}^{n}$. In~\eqref{eq:mutual chain} the chain rule for entropy is used, and (\ref{eq:inf channels}) follows because $\br$ and $\bT$ are independent (see~(\ref{eq:r_k})). In~\eqref{eq:inf channels_2} and~\eqref{eq:inf channels_3}, we used the definition of $\theta_k$, given in~\eqref{eq:theta_k}.

We present the following lemma pertaining to the capacity-achieving input distribution, which will be used throughout the paper.
\begin{lem}
\label{th:uniform_pahse}
The capacity-achieving input of the channel \eqref{eq:channel} is circularly symmetric, i.e., $\{\Theta_k\}_{k=1}^n$ are independently and identically distributed from $\mathcal{U}(0,2\pi)$ and are independent of $\bR$. 
\end{lem}

\begin{proof}[Proof:\nopunct]
The proof directly follows that of~\cite[Prop.~7]{moser09-06a}, where it is shown that the capacity-achieving input of the fading channel with memory is circularly symmetric.
\end{proof}

Based on~(\ref{eq:theta_k}) and the result of Lemma~\ref{th:uniform_pahse}, it can be deduced that the output phase is also uniformly distributed and, hence,~$h(\bt|\br)=n\log_2 2\pi$. 
Thus,~(\ref{eq:inf channels_3}) can be rewritten as
\begin{align}
\frac{1}{n} I(\br,\bt;\bR,\bT)&= \frac{1}{n}\left( I(\br;\bR) - h({\bf N} + {\boldsymbol \phi} | \br,\bR ) \right)+ \log_2 2\pi.\label{eq:inf_uniform_input}
\end{align}

Next, we use the definition of the phase channel~\eqref{eq:theta_k}, and rewrite the second term on the RHS of \eqref{eq:inf_uniform_input} as
\begin{align}
& h( {\bf N} + {\boldsymbol \phi} | \br, \bR)\notag\\
&=h(\ubar{D}( {\bf N} + {\boldsymbol \phi}) | \br, \bR)-\log_2|\ubar{D}|\label{eq:differential}\\
&= h(N_n-N_{n-1}+\Delta_n,\dots,\notag\\&\hspace{3cm}N_2-N_{1}+\Delta_2, N_1+\Delta_1| \br , \bR) \label{eq:differential 0}\\
&= h( \{ N_k-N_{k-1}+\Delta_k \}_{k=2}^n, N_1+\Delta_1 | \br , \bR) \label{eq:differential 1}\\
&=  h \left( \{a_k\}_{k=1}^n| \br,\bR \right) \label{eq:differential 2}
%
\end{align}
where 

\begin{align}
\label{eq:a_definition}
a_k \triangleq
  \begin{cases}
   N_1+\Delta_1				   & \text{if } k = 1 \\
   N_k-N_{k-1}+\Delta_k       & \text{if } k > 1
  \end{cases}.
\end{align}
In~\eqref{eq:differential}, we choose $\ubar{D}$ to be the \emph{difference matrix} defined as 
\begin{align}
\ubar{D}=\left[ \begin{array}{ccccc}
1 & -1 & 0 & \dots &0\\
0 & 1  & -1& \ddots & \vdots\\
\vdots & \ddots  & 1 & \ddots & 0\\
\vdots & \ddots  & \ddots & \ddots & -1\\
0 & \dots  & \dots & 0 & 1\\
\end{array} \right]_{n\times n},
\end{align}
where $|\ubar{D}|=1$. The equality in~\eqref{eq:differential} follows from~\cite[Eq.~8.71]{cover06-a}, and in~\eqref{eq:differential 0}, we used that~$\log_2(|\ubar{D}|)=0$.
The noise vector ${\bf N} + {\boldsymbol \phi}$ is rearranged in~\eqref{eq:differential} as the difference between the consecutive noise samples in order to resolve the infinite memory of the phase-noise process $\phi_k$. 
%
By substituting~\eqref{eq:differential 2} in~\eqref{eq:inf_uniform_input}, we obtain
\begin{align}
\frac{1}{n} I(\br,\bt;\bR,\bT)&= \frac{1}{n}\left( I(\br;\bR) - h \left( \{a_k\}_{k=1}^n| \br,\bR \right) \right)+ \log_2 2\pi.
\label{eq:info_rate_final}
\end{align}
In order to find upper and lower bounds on the capacity, we need to evaluate the two first terms on the RHS of (\ref{eq:info_rate_final}). 
\section{Capacity Upper Bound}\label{sec:upper_bound}
In this section, an upper bound on the capacity of the phase-noise channel \eqref{eq:channel} is derived. We first find a lower bound for the entropy term on the RHS of (\ref{eq:info_rate_final}) as follows
\begin{align}
&h \left( \{a_k\}_{k=1}^n| \br,\bR \right)\notag\\
&\hspace{1.5cm}=\sum_{k=1}^n  h \left( a_k|a_{k-1},\dots,a_1, \br,\bR \right)\label{eq:diff_lower_bound_1}\\
&\hspace{1.5cm}\geq\sum_{k=1}^n  h \left( a_k|a_{k-1},\dots,a_1,N_{k-1}, \br,\bR \right)\label{eq:diff_lower_bound_2}\\
&\hspace{1.5cm}=\sum_{k=1}^n  h \left( a_k|N_{k-1}, \br,\bR \right)\label{eq:diff_lower_bound_3}\\
&\hspace{1.5cm}=\sum_{k=1}^n  h \left( N_k+\Delta_k| r_k,R_k \right)\label{eq:diff_lower_bound_4}\\
&\hspace{1.5cm}=n  h \left( N_n+\Delta_n|r_n,R_n \right)\label{eq:diff_lower_bound_5}.
\end{align}
Here, in~\eqref{eq:diff_lower_bound_1}, the chain rule of differential entropy is used. In~\eqref{eq:diff_lower_bound_2}, we conditioned on the complete knowledge of the noise sample~$N_{k-1}$, and we used that conditioning reduces entropy. Equality in \eqref{eq:diff_lower_bound_3} holds because $a_k$ is conditionally independent of $(a_{k-1},\dots,a_1)$ given~$N_{k-1}$~(see \eqref{eq:a_definition}). Finally, in \eqref{eq:diff_lower_bound_5}, we assumed that $R_k$, and therefore $r_k$, and $N_k+\Delta_k$ are stationary. As the channel is stationarity, one may show that there is no gain in choosing a non-stationary input.

We next upper-bound the mutual information on the RHS of~(\ref{eq:info_rate_final}) by using the \emph{duality approach}~\cite[Th.~5.1]{lapidoth03-10a}. Let $f(\br|\bR)$ denote the conditional probability of~$\br$ given~$\bR$,~and $f(\br)$ denote the distribution of~$\br$ for a given input distribution~$f(\bR)$, and lastly, let~$q(\br)$ be an arbitrary distribution of~$\br$. The mutual information in~\eqref{eq:info_rate_final} can be upper-bounded by using duality as~\cite[Th.~5.1]{lapidoth03-10a}
\begin{align}
	I(\br;\bR)&=\Ex{f(\bR)}{\mathcal{D}\big(f(\br|\bR)||f(\br)\big)}\label{eq:duality_step_1}\\
				 &=\Ex{f(\bR)}{\mathcal{D}\big(f(\br|\bR)||q(\br)\big)}-\mathcal{D}\big(f(\br)||q(\br)\big)\label{eq:duality_step_2}\\
&\leq\Ex{f(\bR)}{\mathcal{D}\big(f(\br|\bR)||q(\br)\big)}\label{eq:duality_step_3}\\
&=-\Ex{f(\br)}{\ln\big( q(\br)\big)}-h(\br|\bR).\label{amplitude_channel_MI}
\end{align}
Here,~$\mathcal{D}(\cdot||\cdot)$ denotes the relative entropy between two probability distributions~\cite[Eq.~8.46]{cover06-a}; \eqref{eq:duality_step_1} follows the definition of the mutual information~\cite[Eq.~8.49]{cover06-a}; in \eqref{eq:duality_step_2} we used Tops\o e's identity \cite{topsoe67-a}; \eqref{eq:duality_step_3} follows because of the nonnegativity of relative entropy~\cite[Thm.~8.6.1]{cover06-a}. Finally,~\eqref{amplitude_channel_MI} follows diretly from the definition of the relative entropy~\cite[Eq.~8.46]{cover06-a}. 

From~\eqref{eq:duality_step_3}, we see that any choice of the auxiliary output distribution~$q(\br)$ results in an upper bound for~$I(\br;\bR)$. However, $q(\br)$ needs to be selected such that a tight upper bound is obtained. Specifically, we choose the output amplitudes to be independently distributed from~$q(r)$, which is a particular mixture of a half-normal distribution and a Rayleigh distribution
\begin{align}
q(r) &= \frac{\alphaU(\mu)}{\sqrt{\frac{\sigN}{(r+\mu)^2}+\sigd}} e^{-\betaU(\mu) r^2},\quad r>0. \label{eq:pdf upper}
\end{align}
We will soon motivate the form of~$q(r)$. Moreover, in Section~\ref{sec:Results}, we will show that this choice of $q(r)$ results in a tight upper bound on the capacity. 

In~\eqref{eq:pdf upper}, $\mu\geq 0$ is a constant that will be optimized later to tighten the upper bound. For any $\mu$, the parameters $\alphaU(\mu)$ and $\betaU(\mu)$ should be chosen such that certain constraints on $q(r)$ are satisfied. 
The first constraint is based on the fact that $q(r)$ is a probability distribution function and, hence, must integrate to one
\begin{align}
\int_0^\infty q(r) \text{d}r= 1 \label{eq:L_U1}.
\end{align}
The second constraint is due to the input power constraint \eqref{eq:power_constraint}, and can be found from \eqref{eq:r_k}
\begin{align}
\int_0^\infty r^2 q(r) \text{d}r &= \Es + 2\sigN.\label{eq:L_U2}
\end{align}

Although finding closed-form expressions of $\alphaU(\mu)$ and $\betaU(\mu)$ is not straightforward, it is possible to determine their values numerically. The numerical method that we used for computing these parameters is presented in Appendix~\ref{sec:numerical_calc_alpha_beta}, and  Tab.~\ref{tab:upper_bound} contains their computed values for~$\mu=0$ and for various values of $\sigN$ and $\sigd$. 

The transition between the half-normal, and the Rayleigh distributions
in~\eqref{eq:pdf upper} is based on the values of $\sigN$ and $\sigd$. At high SNR, where the phase noise dominates~($\sigN<<\sigd$),~$q(r)$ is asymptotically half-normal, while it is a Rayleigh distribution for the low SNR values. 

This choice of auxiliary output distribution in~\eqref{eq:pdf upper} is motivated as follows: i) The capacity-achieving distribution of the Gaussian channel is a normal distribution~\cite{cover06-a}, and thus the input (and also the output) amplitude follows a Rayleigh distribution, ii) As shown in~\cite{lapidoth02-10a}, a tight upper bound for the phase-noise channel at high SNR can be found by using the duality approach, and considering an optimized Gamma distribution as an auxiliary distribution on $|y|^2$. In that case, by following the standard technique for determining the probability density function of a transformed random variable~\cite[Ch.~5]{papoulis2002probability}, it is straightforward to show that~$r=|y|$ follows a half-normal distribution. 
%

%

\begin{table}[t]
\caption{Numerically calculated values of $\alphaU(\mu=0)$ and $\betaU(\mu=0)$, when $\Es=1$, and for various $\sigN$ and $\sigd$.}
\label{tab:upper_bound}
\centering
\renewcommand\arraystretch{1.3}
\begin{tabular}{c|c|c|c|c|}
\cline{2-4}
& $\sigN$& $\alphaU(\mu=0)$ & $\betaU(\mu=0)$ \\ \cline{1-4}
\multicolumn{1}{ |c| }{\multirow{5}{*}{$\sigd=10^{-2}~[\text{rad}^2]$} } &
\multicolumn{1}{  c| } {$5\times 10^{-2}$} & 0.43 & 0.88\\ \cline{2-4}
\multicolumn{1}{ |c|  }{}                        &
\multicolumn{1}{  c| }{$5\times 10^{-3}$} & 0.17 & 0.73\\ \cline{2-4}
\multicolumn{1}{ |c|  }{}                        &
\multicolumn{1}{  c| }{$5\times 10^{-4}$} & 0.10 & 0.59\\ \cline{2-4}
\multicolumn{1}{ |c|  }{}                        &
\multicolumn{1}{  c| }{$5\times 10^{-5}$} & 0.09 & 0.53\\ \cline{2-4}
\multicolumn{1}{ |c|  }{}                        &
\multicolumn{1}{  c| }{$5\times 10^{-6}$} & 0.08 & 0.51\\ \cline{1-4}
\multicolumn{1}{ |c| }{\multirow{5}{*}{$\sigd=10^{-3}~[\text{rad}^2]$} } &
\multicolumn{1}{  c| }{$5\times 10^{-2}$} & 0.43 & 0.94\\ \cline{2-4}
\multicolumn{1}{ |c|  }{}                        &
\multicolumn{1}{  c| }{$5\times 10^{-3}$} & 0.14 & 0.92\\ \cline{2-4}
\multicolumn{1}{ |c|  }{}                        &
\multicolumn{1}{  c| }{$5\times 10^{-4}$} & 0.05 & 0.73\\ \cline{2-4}
\multicolumn{1}{ |c|  }{}                        &
\multicolumn{1}{  c| }{$5\times 10^{-5}$} & 0.03 & 0.59\\ \cline{2-4}
\multicolumn{1}{ |c|  }{}                        &
\multicolumn{1}{  c| }{$5\times 10^{-6}$} & 0.03 & 0.53\\ \cline{1-4}
\end{tabular}
\end{table}

By substituting \eqref{eq:pdf upper} in \eqref{amplitude_channel_MI}, we obtain
\begin{align}
	I(\br;\bR)&\leq-n\Ex{}{\log_2 q(r)}-n h(r|R),\label{amplitude_channel_MI_with_pdf_0}\\				 &=-n\log_2(\alphaU(\mu))+n\frac{\betaU(\mu)}{\ln(2)}\Ex{}{r^2}\notag\\
&\quad+\frac{n}{2}\Ex{}{\log_2\left(\frac{\sigN}{(r+\mu)^2}+\sigd\right)}-n h(r|R),\label{amplitude_channel_MI_with_pdf}
\end{align}
where in \eqref{amplitude_channel_MI_with_pdf_0}, we used that $r$ is independently and identically distributed (iid). From \eqref{eq:L_U2}, we have~$\Ex{}{r^2}=\Es + 2\sigN$. By substituting \eqref{eq:diff_lower_bound_5} and \eqref{amplitude_channel_MI_with_pdf} in \eqref{eq:info_rate_final}, then~\eqref{eq:info_rate_final} in~\eqref{eq:capacity def}, we obtain an upper bound on the capacity as
\begin{align}
C(\SNR) \leq &-\log_2\left(\frac{\alphaU(\mu)}{2\pi}\right)+\frac{\betaU(\mu)}{\ln(2)}(\Es + 2\sigN)\notag\\&
+\sup_{f(R)}\Bigg\{\frac{1}{2}\Ex{}{\log_2\left(\frac{\sigN}{(r+\mu)^2}+\sigd\right)}\notag\\
&\quad\quad\quad\quad-h(r|R) - h\left( N +\Delta | r , R \right)\Bigg\},\label{eq:upper_bound_before_ecape} \\ 
=&-\log_2\left(\frac{\alphaU(\mu)}{2\pi}\right)+\frac{\betaU(\mu)}{\ln(2)}(\Es + 2\sigN)\notag\\&
+\sup_{f(R)}\Bigg\{\mathbb{E}_{f(R)}\Bigg[\frac{1}{2}\Ex{f(w_{\parallel},w_{\perp})}{\log_2\left(\frac{\sigN}{(r+\mu)^2}+\sigd\right)}\notag\\
&\quad\quad\quad\quad-h(r|R) - h\left( N +\Delta | r , R \right)\Bigg]\Bigg\}.\label{eq:upper_bound_before_ecape_clarity1}  
\end{align}
Finally, the capacity of the phase-noise channel~\eqref{eq:channel} can be bounded as~$C(\SNR)\leq C_\text{U}(\SNR)$, where
\begin{align}
C_\text{U}(\SNR) = \min_{\mu\geq 0}\Big\{
-\log_2\left(\frac{\alphaU(\mu)}{2\pi}\right)&+\frac{\betaU(\mu)}{\ln(2)}(\Es + 2\sigN)\notag\\&+\max_{R\geq 0} \mathcal{G}(R)\Big\},\label{eq:upper_bound_before_ecape_2}  
\end{align}
and  
\begin{align}
&\mathcal{G}(R) = \notag\\&\hspace{0.5cm}\frac{1}{2}\Ex{f(w_{\parallel},w_{\perp})}{\log_2\left(\frac{\sigN}{\left(\sqrt{(R + w_{\parallel})^2 + w^2_{\perp}}+\mu\right)^2}+\sigd\right)}\notag\\
				&\hspace{0.5cm}-h\left(\sqrt{(R + w_{\parallel})^2 + w^2_{\perp}}\right)\notag\\
				&\hspace{0.5cm}-h\left(\arctan \frac{w_{\perp}}{R+w_{\parallel}} +\Delta \Big | r\right).\label{eq:definition_of_g_before_escape}
\end{align} 
In~\eqref{eq:upper_bound_before_ecape_2}, the expectation over~$f(R)$ is upper-bounded by the maximum value of the expression, which expectation is taken over, for a given $R$. Finally, the bound can be tightened by minimizing over $\mu\geq 0$.

The upper bound in~\eqref{eq:upper_bound_before_ecape} is further simplified and the final result is presented in the following proposition.
\begin{prop}
\label{prop:upper_bound}
Capacity of the Wiener phase-noise channel~\eqref{eq:channel} is upper-bounded as
\begin{align}
\label{eq:upper_bound_escape_definition_SNR}
 C(\SNR) \leq C_{\tilde{\text{U}}}(\SNR) + o(1),\quad \SNR\to\infty 
\end{align}
where $o(1)$ denotes a function that vanishes as SNR grows large, and
\begin{align}\label{eq:final_upBound}
&C_{\tilde{\text{U}}}(\SNR)=\notag\\ 
&\hspace{0.7cm}\frac{\betaU(\mu=0)}{\ln(2)}(\Es +2\sigN) - \frac{1}{2}\log_2\sigN e^2 \alphaU^2(\mu=0).
\end{align}
\end{prop}

\begin{proof}
Please refer to Appendix~\ref{sec:upper_bound_proof_continued}.
\end{proof}

As we shall see from simulation results in Section~\ref{sec:Results}, the provided upper bound is tight for a wide range of SNR values. 

\section{Capacity Achieving Distributions and the Capacity Lower Bound}\label{sec:lower_bound}
%
One approach to find the capacity lower bound is to restrict the input to have a particular distribution. However, the input distribution must be chosen such that a tight lower bound is obtained. In this section, we present a method to intelligently choose the distribution of input amplitudes, $f(\bR)$. 

We first reconsider the amplitude and phase channel models in \eqref{eq:r_k} and \eqref{eq:theta_k}, which at high SNR, reduce to
\begin{align}
r_k &= \sqrt{(R_k + w_{k,\parallel})^2 + w^2_{k,\perp}} \\
&= \left( R_k + w_{k,\parallel} \right) \sqrt{1 + \frac{w^2_{k,\perp}}{(R_k + w_{k,\parallel})^2}} \\
&\approx R_k + w_{k,\parallel}\label{eq:r_k approx der}\\
\theta_k &=\Theta_k + \arctan \frac{w_{k,\perp}}{R_k+w_{k,\parallel}} + \phi_k\\
&\approx \Theta_k + N_k + \phi_k, \label{eq:theta_k approx der}
\end{align}
where $N_k\triangleq{w_{k,\perp}}/{r_k}$. In (\ref{eq:theta_k approx der}), we used that $\arctan(z) \approx  z$ for small $z$. In the following, we study the channel defined in \eqref{eq:r_k approx der} and \eqref{eq:theta_k approx der}, and derive the capacity-achieving distribution for this simplified channel. Note that in this section, $r_k$, $\theta_k$, and $N_{k}$ refer to the parameters of the approximate channel to avoid defining new variables.

By using the approximate input-output amplitude and phase relations in \eqref{eq:r_k approx der} and \eqref{eq:theta_k approx der}, and by following similar steps that lead to \eqref{eq:info_rate_final}, we obtain

\begin{IEEEeqnarray}{rCL}
&&\hspace{-0.7cm}\frac{1}{n} I(\br,\bt;\bR,\bT)\notag\\ 
&&\hspace{-0.7cm}= \frac{1}{n}\big( I(\br;\bR) - h \left( \{a_k\}_{k=1}^n| \br,\bR \right) \big)+ \log_2 2\pi\label{eq:info_rate_approx_st01}\\
&&\hspace{-0.7cm}= \frac{1}{n}\big( h(\br)- h(\br|\bR) - h \left( \{a_k\}_{k=1}^n| \br\right) \big)+ \log_2 2\pi\label{eq:info_rate_approx_st2}\\
%
%
&&\hspace{-0.7cm}= \frac{1}{n}\big( h(\br)- h(\mathbf{w}_\parallel) - h \left( \{a_k\}_{k=1}^n| \br \right) \big)+ \log_2 2\pi\\
%
%
&&\hspace{-0.7cm}=\frac{1}{n}\big(h(\br) - h \left( \{a_k\}_{k=1}^n| \br \right)\big)\notag\\&&\hspace{3cm}-\frac{1}{2}\log_2 2\pi e \sigN  + \log_2 2\pi, \label{eq:info_rate_final_approx}
\end{IEEEeqnarray}
where $\mathbf{w}_\parallel=\{w_{k,\parallel}\}_{k=1}^n$. In~\eqref{eq:info_rate_approx_st2}, the first term of~\eqref{eq:info_rate_approx_st01} is rewritten by using the definition of mutual information. For the second term, we used that~$\{a_k\}_{k=1}^n$, given~$\br$, is independent of~$\bR$ because of~\eqref{eq:theta_k approx der}. In~(\ref{eq:info_rate_final_approx}), we used~\eqref{eq:r_k approx der} and that of the entropy of the~$n$-dimensional Gaussian-distributed random variable~$\mathbf{w}_\parallel$ is given by \cite[Thm.~8.4.1]{cover06-a}
\begin{align}
h(\mathbf{w}_\parallel)=\frac{n}{2}\log_2 2\pi e \sigN.
\end{align}

Any choice of $f(\bR)$ results in a lower bound on the capacity of the approximate amplitude-phase channel. Note that the infinite memory of the Wiener phase-noise process, $\phi_k$, is resolved by rearranging the noise vector ${\bf N} + {\boldsymbol \phi}$ into the difference between the consecutive noise samples. This results in dependency among $\{a_k\}_{k=1}^{n}$ samples (see~\eqref{eq:differential}-\eqref{eq:differential 2}), and motivates to consider an input distribution that introduces a limited order dependency across the amplitudes of the consecutive symbols. More specifically, we consider \emph{block-independent} input amplitudes and confine the optimization in~\eqref{eq:capacity_polar_def} to the set of input distributions of the form
\begin{align}
\label{eq:pdf_limited_order_dep_input_set}
f(\bR)=\prod_{k=1}^{n/M} f\left({\Rt}^{(k)}\right)=\left(f({\Rt})\right)^{n/M},
\end{align}
where~$\Rt^{(k)}$ are blocks of length-$M>1$ samples obtained by dividing the vector of input amplitudes \footnote{Length $M=1$ is an uninteresting case because $a_n$ depends at least on two consecutive symbols.}
\begin{align}
\hspace{-0.5cm}\bR= \underbrace{\{R_1,\dots,R_M\}}_{\triangleq\Rt^{(1)}},\underbrace{\{R_{M+1},\dots,R_{2M}\}}_{\triangleq\Rt^{(2)}},\dots,
 \underbrace{\{R_{n-M+1},\dots,R_n\}}_{\triangleq\Rt^{(n/M)}}.
\end{align}
The second equality in~\eqref{eq:pdf_limited_order_dep_input_set} follows as~${\Rt}^{(k)}$ are iid from $f(\Rt)$. 
%
%
By substituting~\eqref{eq:info_rate_final_approx} in~\eqref{eq:capacity_polar_def}, and limiting the input distributions to~\eqref{eq:pdf_limited_order_dep_input_set}, we obtain
\begin{align}
C (\SNR) \geq  \lim_{n \to \infty} \underset{f(\Rt)}{\sup} \; \: \left\{\frac{1}{n}\big(h(\br) - h \left( \{a_k\}_{k=1}^n| \br \right)\big)\right\}\notag \\+ \log_2 2\pi-\frac{1}{2}\log_2 2\pi e \sigN.\label{eq:capacity_lower_bound_indep_input}
\end{align}
According to the input-output relation \eqref{eq:r_k approx der}, block-independent input symbols result in block-independent output samples denoted by
\begin{align}
\br= \underbrace{\{r_1,\dots,r_M\}}_{\triangleq\rt^{(1)}},\underbrace{\{r_{M+1},\dots,r_{2M}\}}_{\triangleq\rt^{(2)}},\dots, \underbrace{\{r_{n-M+1},\dots,r_n\}}_{\triangleq\rt^{(n/M)}}.\label{eq:grouped_output}
\end{align}
Thus, a lower bound on the capacity can be found by using~\eqref{eq:capacity_lower_bound_indep_input} and~\eqref{eq:grouped_output} as
\begin{align}
C (\SNR)&\notag\\
&\hspace{-1.4cm}\geq\lim_{n \to \infty} \underset{f(\Rt)}{\sup} \bigg\{\frac{1}{n}\left( h\left(\rt^{(1)},\rt^{(2)},\ldots,\rt^{(n/M)}\right)-  h \left(\{a_k\}_{k=1}^{n}|\br\right)\right)\bigg\}\notag\\ 
&\hspace{2.2cm}-\frac{1}{2}\log_2 2\pi e \sigN + \log_2 2\pi\label{eq:inf lower_bound_2}\\
&\hspace{-1.4cm}=\lim_{n \to \infty} \underset{f(\Rt)}{\sup} \bigg\{\frac{1}{M}h(\rt^{(n/M)})- \frac{1}{n} h \left(\{a_k\}_{k=1}^{n}|\br\right)\bigg\}\notag\\ 
&\hspace{2.2cm}-\frac{1}{2}\log_2 2\pi e \sigN + \log_2 2\pi\label{eq:inf lower_bound_3}\\
&\hspace{-1.4cm}\geq\underset{f(\Rt)}{\sup} \bigg\{\frac{1}{M}h(\rt^{(n/M)})- \lim_{n \to \infty}\frac{1}{n} h \left(\{a_k\}_{k=1}^{n}|\br\right)\bigg\}\notag\\ 
&\hspace{2.2cm}-\frac{1}{2}\log_2 2\pi e \sigN + \log_2 2\pi\label{eq:inf lower_bound_4}\\
&\hspace{-1.4cm}=\underset{f(\Rt)}{\sup} \bigg\{\frac{1}{M}h(\rt^{(n/M)})- \lim_{n \to \infty} h \left(a_n|\{a_k\}_{k=1}^{n-1},\br\right)\bigg\}\notag\\ 
&\hspace{2.2cm}-\frac{1}{2}\log_2 2\pi e \sigN + \log_2 2\pi\label{eq:inf lower_bound_51}\\
&\hspace{-1.4cm}\geq\underset{f(\Rt)}{\sup} \bigg\{\frac{1}{M}h(\rt^{(n/M)})\notag\\&-\lim_{n \to \infty} h \left(a_n|\{a_k\}_{k=n-M+2}^{n-1},\rt^{(n/M)}\right)\bigg\}\notag\\ 
&\hspace{2.2cm}-\frac{1}{2}\log_2 2\pi e \sigN + \log_2 2\pi.\label{eq:inf lower_bound_61}
\end{align}
Here, in \eqref{eq:inf lower_bound_3}, we used that $\rt^{(k)}$ are iid and stationary, and then rewrote the first term inside the supremum. In \eqref{eq:inf lower_bound_4} we swapped the supremum and the limit operations, which resulted in the inequality. Appendix~\ref{sec:sup_lim_swap} describes the steps involved in the swapping. In~\eqref{eq:inf lower_bound_51}, the joint entropy is replaced with an equivalent expression for the differential-entropy rate~\cite[Th.~4.2.1]{cover06-a}.
The inequality in \eqref{eq:inf lower_bound_61} holds because removing conditioning from the second term increases entropy. By using the chain rule of entropy on the second term of~\eqref{eq:inf lower_bound_61}, we obtain
\begin{align}
C (\SNR)
&\geq\underset{f(\Rt)}{\sup}\bigg\{\frac{1}{M}h(\rt)\notag\\
&\quad\quad\quad\quad+h \left( \{a_k\}_{k=n-M+2}^{n-1}|\rt \right)\notag\\
&\quad\quad\quad\quad-h \left( a_n,\{a_k\}_{k=n-M+2}^{n-1}|\rt\right)\bigg\}\notag\\ 
&\quad-\frac{1}{2}\log_2 2\pi e \sigN+ \log_2 2\pi,\label{eq:inf lower_bound_5}
\end{align}
where the superscript of $\rt$, and the limit are removed due to stationarity. According to~\eqref{eq:theta_k approx der}, the second and the third terms on the RHS of \eqref{eq:inf lower_bound_5} correspond to the entropy of two zero-mean Gaussian random vectors with covariance matrices defined as
\begin{IEEEeqnarray}{rCL}
\label{eq:an_covs}
\Sigma_{n-1} &\triangleq 
\text{cov}(\{a_k\}_{k=n-M+2}^{n-1}|\rt)\IEEEyesnumber\IEEEyessubnumber\\
\Sigma_{n}\quad &\triangleq \text{cov}(\{a_k\}_{k=n-M+2}^{n}|\rt)\IEEEyessubnumber.
\end{IEEEeqnarray} 
By using  \eqref{eq:an_covs} and the definition of the entropy for a Gaussian random variable \cite[Thm.~8.4.1]{cover06-a}, \eqref{eq:inf lower_bound_5} can be rewritten as 
\begin{align}
C (\SNR)
&\geq\underset{f(\Rt)}{\sup}\bigg\{
\frac{1}{M}h(\rt)\notag\\
&\quad\quad\quad\quad-\Ex{f(\rt)}{\frac{1}{2}\log_2 2\pi e~g_\br(\rt)}\bigg\}\notag\\ 
&\quad+ \log_2 2\pi-\frac{1}{2}\log_2 2\pi e \sigN,\label{eq:inf lower_bound_6}
\end{align}
where $g_\br(\rt)={|\Sigma_n|}/{|\Sigma_{n-1}|}$.

In order to find a tight lower bound on the capacity, we need to find the supremum in~\eqref{eq:inf lower_bound_6} by searching over all probability distributions on $\Rt$ that satisfy the power constraint~\eqref{eq:power_constraint}. Unfortunately, this optimization is not mathematically tractable in general. However, the search space reduces at high SNR, since~$f(\rt)$ converges to~$f(\Rt)$. Here, determining the optimized~$f(\Rt)$ becomes equivalent to finding the optimized~$f(\rt)$. Therefore, we maximize~\eqref{eq:inf lower_bound_6} over~$f(\rt)$ by means of functional optimization. The optimized~$f(\rt)$ is used as~$f(\Rt)$, which is used to evaluate the lower bound for the capacity of the channel in~\eqref{eq:channel}. The optimization steps needed are described in Appendix~\ref{sec:functional_optimization}, and finally,~$f(\Rt)$ is found as
\begin{align}
f(\Rt) &= \alphaL^{(M)} \left(g_{\br}(\Rt)\right)^{-M/2} e^{-\left(\betaL^{(M)} ||\Rt||^2\right)},\label{eq:pdf lower R}
\end{align}
where the parameters~$\alphaL^{(M)}$ and~$\betaL^{(M)}$ are chosen such that the following constraints are satisfied. The first constraint is based on the fact that $f(\Rt)$ is a probability distribution function and, hence, must integrate to one
\begin{align}
\int_0^\infty f(\Rt) \text{d}\Rt= 1 \label{eq:L_L1_R}.
\end{align}
The second constraint is due to the input power constraint \eqref{eq:power_constraint} 
\begin{align}
\int_0^\infty ||\Rt||^2 f(\Rt) \text{d}\Rt &= M\Es .\label{eq:L_L2_R}
\end{align}
%
%
%
To numerically compute $\alphaL^{(M)}$ and~$\betaL^{(M)}$, a method similar to that presented in Appendix~\ref{sec:numerical_calc_alpha_beta} is used. Tab.~\ref{tab:lower_bound} contains the calculated values of these parameters in various scenarios. Note that the superscript $(M)$ shows the dependency of these parameters on the dimension of~$\Rt$. 

In the numerical-result section, we will see that choosing the input distribution as proposed in~\eqref{eq:pdf lower R} results in a tight lower bound on the capacity for a wide range of SNR values and phase noise innovation variances. 
\begin{table}[t]
\caption{Numerically calculated values of $\alphaL$ and $\betaL$, for~$\Es=1$, $\sigd=10^{-3}$ and various $\sigN$.}
\label{tab:lower_bound}
\centering
\renewcommand\arraystretch{1.3}
\begin{tabular}{|c|c|c|c|c|c|}
\cline{1-5}
$\sigN$& $\alphaL^{(2)}$ & $\betaL^{(2)}$ & $\alphaL^{(3)}$ & $\betaL^{(3)}$\\ \cline{1-5}
$5\times 10^{-2}$ & $0.509$ & $0.997$ & $0.12500$ & $0.991$\\ \cline{1-5}
$5\times 10^{-3}$ & $0.051$ & $0.967$ & $0.00400$ & $0.936$\\ \cline{1-5}
$5\times 10^{-4}$ & $0.006$ & $0.825$  & $0.00020$ & $0.756$\\ \cline{1-5}
$5\times 10^{-5}$ & $0.001$ & $0.634$ & $0.00003$ & $0.598$\\ \cline{1-5}
$5\times 10^{-6}$ & $0.001$ & $0.544$ & $0.00002$ & $0.533$\\ \cline{1-5}
\end{tabular}
\end{table}
%

In the following, we evaluate $f(\Rt)$ for ~$M=2$ and~$M=3$, which will be used in our numerical simulations. For $M=2$, the argument of the limit in \eqref{eq:inf lower_bound_61} is equal to~$h \left( a_n|\rt \right)$. Hence,  
\begin{align}
\label{eq:2d_gr}
g_{\br}(\rt) = |\Sigma_n|&=\text{cov}(a_n|\rt)\notag\\
&=\sigd + \frac{\sigN}{r^2_{n}} + \frac{\sigN}{r^2_{n-1}}.
 \end{align}
By using \eqref{eq:pdf lower R} and \eqref{eq:2d_gr}, we obtain the  two-dimensional input distribution as
\begin{align}
f(R_n,R_{n-1}) &= \alphaL^{(2)} \frac{e^{-\left(\betaL^{(2)} (R_n^2+R_{n-1}^2)\right)}}{\sigd + \frac{\sigN}{R^2_{n}} + \frac{\sigN}{R^2_{n-1}}} . \label{eq:pdf lower R 2D}
\end{align}

For $M=3$, the covariance matrices in \eqref{eq:an_covs} can be computed as
\begin{align}
\label{eq:an_covs_3D}
&\Sigma_{n-1} =\text{cov}(a_{n-1}|\rt)=\sigd + \frac{\sigN}{r^2_{n-1}} + \frac{\sigN}{r^2_{n-2}},\\
&\Sigma_{n}\quad =\text{cov}(a_n,a_{n-1}|\rt)\\
&\qquad~=\begin{pmatrix}
\sigd + \frac{\sigN}{r^2_n} + \frac{\sigN}{r^2_{n-1}} & \frac{\sigN}{r^2_{n-1}} \\
\frac{\sigN}{r^2_{n-1}} & \sigd + \frac{\sigN}{r^2_{n-1}} + \frac{\sigN}{r^2_{n-2}} \\
\end{pmatrix}
.
\end{align}
By using \eqref{eq:pdf lower R} we obtain a three-dimensional input distribution
\begin{align}
f(R_n,R_{n-1},R_{n-2}) &= \alphaL^{(3)}  \frac{e^{-\left(\betaL^{(3)} (R_n^2+R_{n-1}^2+R_{n-2}^2)\right)}}{\left(g_{\br}(\Rt)\right)^{3/2}}, \label{eq:pdf lower R 3D}
\end{align}
where 
\begin{align}
g_{\br}(\Rt)
&=\sigd + \frac{\sigN}{R^2_n} + \frac{\sigN}{R^2_{n-1}}-\frac{\frac{\sigma^4_\text{W}}{R^4_{n-1}}}{\sigd + \frac{\sigN}{R^2_{n-1}} + \frac{\sigN}{R^2_{n-2}}}.
\label{eq:g_R_3D}
\end{align}

\section{Numerical Results}
\label{sec:Results}
\begin{figure}[t]
\begin{center}
\psfrag{LPND}[][][0.8]{$10\log(1/\sigd)$}
\psfrag{lUpBound}[][][0.8]{$C_\text{U}$~(Eq.~\ref{eq:upper_bound_before_ecape_2})}
\psfrag{lEUpBound}[][][0.8]{$C_{\tilde{\text{U}}}$~(Eq.~\ref{eq:final_upBound})}
\psfrag{Lawgn}[][][0.8]{$C_\text{AWGN}$~(Eq.~\ref{eq:AWGN_capacity})}
\psfrag{LLapidoth}[][][0.8]{$C_\text{Lapidoth}$~(Eq.~\ref{eq:Lapidoth_capacity})}
\psfrag{LBound2}[][][0.6]{Lower Bound $M=2$}
\psfrag{LBound3}[][][0.6]{Lower Bound $M=3$}
\includegraphics[width=3.4in]{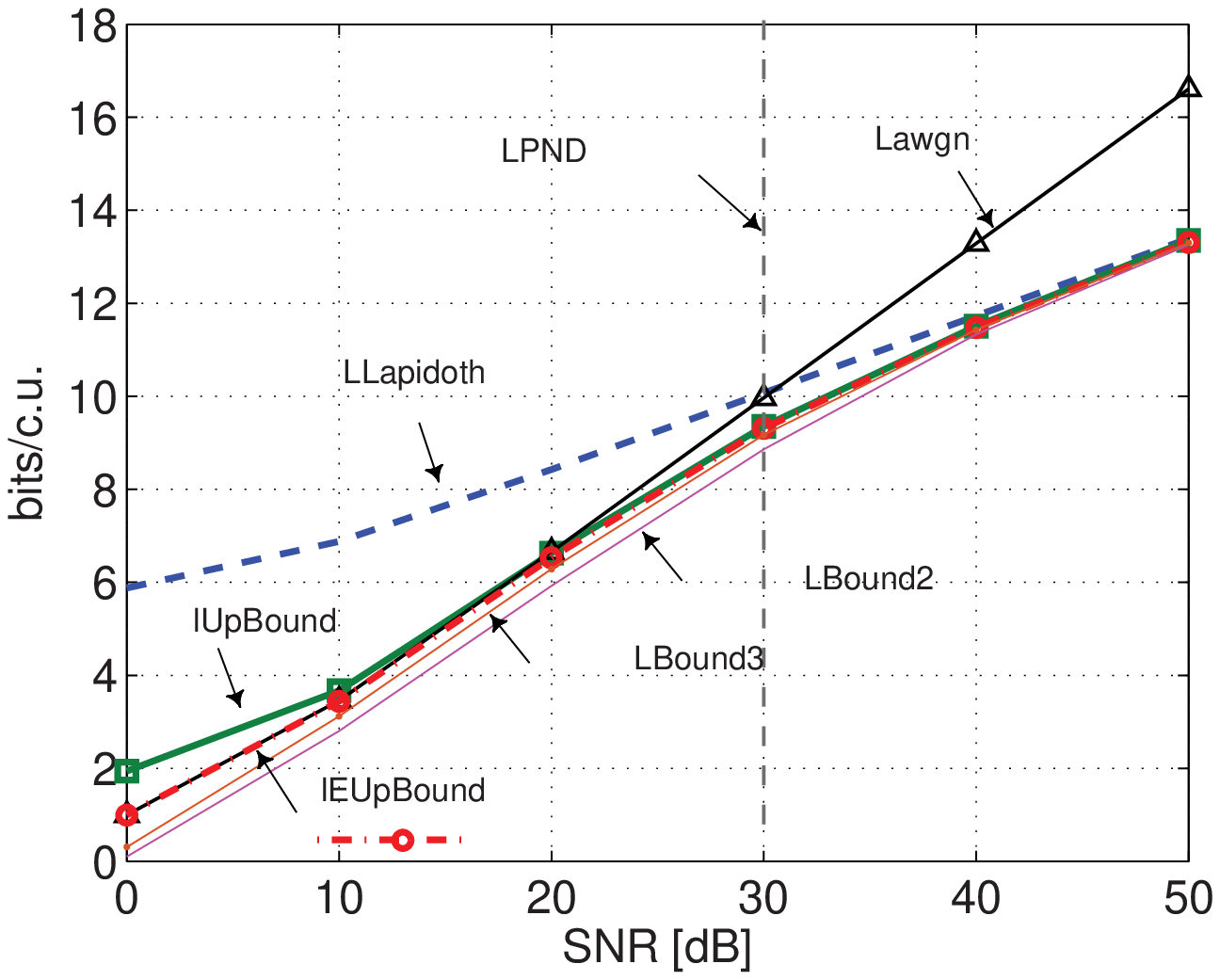}
\caption{The proposed capacity bounds vs. the capacity of the AWGN channel and the asymptotic high-SNR capacity of the phase-noise channel from \cite{lapidoth02-10a}. Here,  $\sigd=10^{-3}~[\text{rad}^2]$.}
\label{fig:UpBound1e-3}
\end{center}
\end{figure}

\begin{figure}[t]
\begin{center}
\psfrag{LPND}[][][0.8]{$10\log(1/\sigd)$}
\psfrag{lUpBound}[][][0.8]{$C_\text{U}$~(Eq.~\ref{eq:upper_bound_before_ecape_2})}
\psfrag{lEUpBound}[][][0.8]{$C_{\tilde{\text{U}}}$~(Eq.~\ref{eq:final_upBound})}
\psfrag{Lawgn}[][][0.8]{$C_\text{AWGN}$~(Eq.~\ref{eq:AWGN_capacity})}
\psfrag{LLapidoth}[][][0.8]{$C_\text{Lapidoth}$~(Eq.~\ref{eq:Lapidoth_capacity})}
\psfrag{LBound2}[][][0.6]{Lower Bound $M=2$}
\psfrag{LBound3}[][][0.6]{Lower Bound $M=3$}
\includegraphics[width=3.4in]{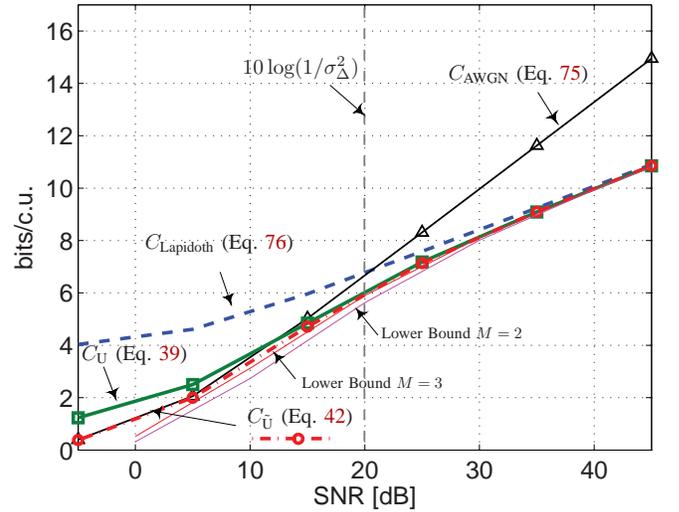}
\caption{The proposed capacity bounds vs. the capacity of the AWGN channel and the asymptotic high-SNR capacity of the phase-noise channel from \cite{lapidoth02-10a}. Here,  $\sigd=10^{-2}~[\text{rad}^2]$.}
\label{fig:UpBound1e-2}
\end{center}
\end{figure}

%

In this section, we first evaluate the proposed upper bounds~\eqref{eq:upper_bound_before_ecape_2} and~\eqref{eq:upper_bound_escape_definition_SNR}. Then, we use the derived input distribution~\eqref{eq:pdf lower R} to find lower bounds on the capacity through numerical simulations. We compare the bounds with the capacity of the 
AWGN channel \cite[Ch.~9]{cover06-a}
\begin{align}
\label{eq:AWGN_capacity}
C_\text{AWGN}(\SNR)=\log_2\left(1+\frac{\Es}{2\sigN}\right),
\end{align}
and the asymptotic high-SNR capacity of the Wiener phase-noise channel derived by Lapidoth in \cite{lapidoth02-10a}, 
\begin{align}
\label{eq:Lapidoth_capacity}
C_\text{Lapidoth}(\SNR)=\frac{1}{2}\log_2\left(1+\frac{\Es}{4\sigN}\right)-\frac{1}{2}\log_2\left(\frac{ e \sigd}{2\pi}\right).
\end{align}

Figs.~\ref{fig:UpBound1e-3} and~\ref{fig:UpBound1e-2} illustrate the upper bounds~\eqref{eq:upper_bound_before_ecape_2} and \eqref{eq:final_upBound}, and simulated lower bounds for different SNR values. In Fig.~\ref{fig:UpBound1e-3}, the phase noise innovation variance is  $\sigd=10^{-3}$, while it is $\sigd=10^{-2}$ in Fig.~\ref{fig:UpBound1e-2}. 

The upper bound~$C_\text{U}(\SNR)$ in~\eqref{eq:upper_bound_before_ecape_2} is calculated by means of Monte-Carlo simulations. More specifically, we compute $\mathcal{G}(R)$ in~\eqref{eq:definition_of_g_before_escape} for given values of $R$ and $\mu$, by drawing samples from~$w_{\parallel}$,~$w_{\perp}$, and~$\Delta$.
The first term of $\mathcal{G}(R)$ is a double integral, computed numerically. The differential entropy terms of $\mathcal{G}(R)$ are estimated by using the~\emph{nearest neighbor estimator}~\cite{beirlant1997nonparametric}. 

To evaluate the upper bound in~\eqref{eq:upper_bound_escape_definition_SNR}, we omit the~$o(1)$ term and plot~$C_{\tilde{\text{U}}}(\SNR)$ from \eqref{eq:final_upBound}. Figs.~\ref{fig:UpBound1e-3} and~\ref{fig:UpBound1e-2} show that the this asymptotic upper bound expression~$C_{\tilde{\text{U}}}(\SNR)$ matches~$C_\text{U}(\SNR)$ for SNR values around $10$~dB and above.

It can be seen from Figs.~\ref{fig:UpBound1e-3} and~\ref{fig:UpBound1e-2} that the upper bound~$C_\text{U}(\SNR)$ is not tight for SNRs below~$10$~dB (because it exceeds the AWGN capacity). An alternative upper bound, tighter than $C_\text{U}(\SNR)$, is $\min\{C_\text{U}(\SNR),C_\text{AWGN}(\SNR)\}$. We observe from the figures that the asymptotic upper bound~$C_{\tilde{\text{U}}}(\SNR)$ is a very accurate approximation of this alternative bound.

In general, at low SNR, the capacity upper bound approaches the AWGN capacity~\eqref{eq:AWGN_capacity} because the AWGN dominates over the phase noise. At high SNR, where the phase noise dominates, the derived upper bounds follow the high-SNR capacity of the Wiener phase-noise channel \eqref{eq:Lapidoth_capacity}. It can be seen from Figs.~\ref{fig:UpBound1e-3} and~\ref{fig:UpBound1e-2} that phase noise starts to dominate for SNR values larger than approximately $10\log(1/\sigd)$~dB. The exact value of this point can also be analytically found by intersecting \eqref{eq:AWGN_capacity} and \eqref{eq:Lapidoth_capacity}.


In order to numerically find lower bounds on the capacity of the channel~\eqref{eq:channel}, the sum-product algorithm proposed in~\cite{arnold06-08a} for calculation of the information rate of channels with memory is used. We specifically use the particle-based implementation of this method, which is introduced in~\cite{dauwels2008computation}. First, we use the rejection sampling approach~\cite{robert2004monte} to draw samples from \eqref{eq:pdf lower R} for the input amplitudes. For the phase of the input samples, we use Lemma~\ref{th:uniform_pahse}, and draw independent samples from~$\mathcal{U}(0,2\pi)$. The generated input samples are transmitted over the original channel \eqref{eq:channel}, and the achievable information rate is computed as explained in~\cite{dauwels2008computation}. 
In figs.~\ref{fig:UpBound1e-3} and~\ref{fig:UpBound1e-2}, the simulated lower bounds with~$M=2$ and~$M=3$ are compared against the proposed upper bounds. The particle-based method of~\cite{dauwels2008computation} with~$10^7$ particles over~$10^3$ channel uses is employed. It can be seen that the computed lower bounds are close to the upper bound for a wide range of SNR values. In particular, the input with a higher order of dependency of amplitudes $(M=3)$ results in a tighter lower bound for lower SNRs. 

Fig.~\ref{fig:UpBoundvsSd} shows $C_{\tilde{\text{U}}}(\SNR)$ for different values of phase noise innovation variance, $\sigd$. As expected, the bound approaches the capacity of the AWGN channel without phase noise when $\sigd << \sigN$.
\section{Conclusion}
In this paper, we presented methods to develop tight upper and lower bounds on the capacity of the Wiener phase-noise channel. A capacity upper bound, tighter than that of available in the literature, is derived. We also derived analytical expressions for a family of input distributions, which result in tight lower bounds on the capacity. The proposed input distributions are circularly symmetric and non-Gaussian. Moreover, the input amplitudes are correlated over time. The proposed upper and lower bounds tightly enclose the channel capacity for a wide range of SNR values. The proposed bounds reach the AWGN capacity at low SNR. In the limiting regime of high SNR, only the amplitude of the transmitted signal can be perfectly recovered, whereas the phase is lost. Therefore in that regime, by increasing the SNR gains in capacity can only be achieved through the amplitude channel. 
\begin{figure}[t]
\begin{center}
\psfrag{lUpBound4}[][][0.8]{$C_{\tilde{\text{U}}}$~for~$\sigd=10^{-4}$}
\psfrag{lUpBound3}[][][0.8]{$C_{\tilde{\text{U}}}$~for~$\sigd=10^{-3}$}
\psfrag{lUpBound2}[][][0.8]{$C_{\tilde{\text{U}}}$~for~$\sigd=10^{-2}$}
\psfrag{awgn}[][][0.8]{$C_\text{AWGN}$~(Eq.~\ref{eq:AWGN_capacity})}
\includegraphics[width=3.2in]{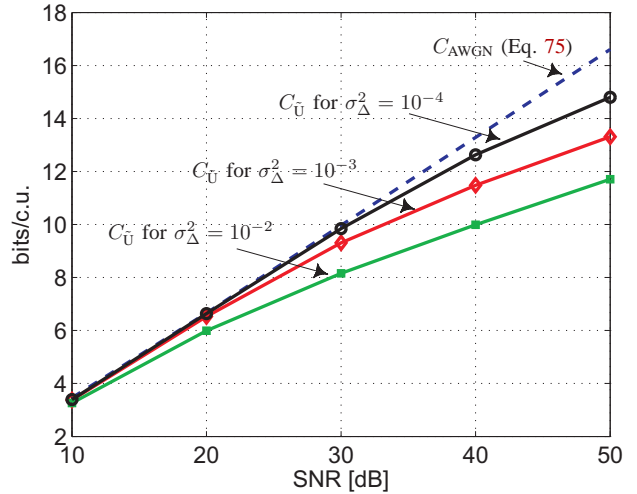}
\caption{The proposed capacity upper bound for various $\sigd$, compared against the AWGN capacity.}
\label{fig:UpBoundvsSd}
\end{center}
\end{figure}
%



\begin{appendices}
\section{Numerical Calculation of \texorpdfstring{$\alphaU(\mu)$}{Lg} and \texorpdfstring{$\betaU(\mu)$}{Lg}}
\label{sec:numerical_calc_alpha_beta}
In this section, we present the numerical method that is used for the calculation of $\alphaU(\mu)$ and $\betaU(\mu)$. From~\eqref{eq:L_U1} and~\eqref{eq:L_U2}, we obtain
\begin{align}
\label{eq:appendix_alpha_beta_int_ratio}
\frac{\int_0^\infty r^2 q(r) \text{d}r}{\int_0^\infty q(r) \text{d}r}= \Es + 2\sigN,
\end{align}
where the left hand side of~\eqref{eq:appendix_alpha_beta_int_ratio} is independent of $\alphaU(\mu)$ based on~\eqref{eq:pdf upper}. To find~$\betaU(\mu)$ that satisfies~\eqref{eq:appendix_alpha_beta_int_ratio}, we use the bisection method. In each iteration of the algorithm, the integrals involved in \eqref{eq:appendix_alpha_beta_int_ratio} are numerically calculated. After finding~$\betaU(\mu)$,~\eqref{eq:L_U1} is used to compute~$\alphaU(\mu)$.
\section{Proof of Proposition~\ref{prop:upper_bound}}
\label{sec:upper_bound_proof_continued}
To further simplify the upper bound in~\eqref{eq:upper_bound_before_ecape_2}, we use the \emph{escape-to-infinity} property of the capacity-achieving input distribution~\cite[Def.~4.11]{lapidoth03-10a}. We impose the additional constraint~$R\geq R_0$ on the input distribution, where~$R_0>0$, and denote the capacity of the channel \eqref{eq:channel} as~$C^{R_0}(\SNR)$. By following the same procedure as in~\cite[Th.~4.12]{lapidoth03-10a} and~\cite[Th.~8]{lapidoth06-02a}, we can show that 
\begin{align}\label{eq:escape_to_infinity}
C(\SNR) = C^{R_0}(\SNR) + o(1),\quad \SNR\to\infty,
\end{align}
where $o(1)$ denotes a function that vanishes as SNR grows large. This means that the high-SNR behavior of $C(\SNR)$ does not change if the input distribution is constrained to lie outside a sphere of an arbitrary radius. Considering this new constraint and repeating the steps leading to \eqref{eq:upper_bound_before_ecape_2}, we obtain 
\begin{align}
C^{R_0}(\SNR) \leq \min_{\mu\geq 0}\Big\{
-\log_2\left(\frac{\alphaU(\mu)}{2\pi}\right)&+\frac{\betaU(\mu)}{\ln(2)}(\Es + 2\sigN)\notag\\&+\max_{R\geq R_0} \mathcal{G}(R)\Big\},\label{eq:upper_bound_ecape_0}  
\end{align}
where~$\mathcal{G}(R)$ is defined in~\eqref{eq:definition_of_g_before_escape}. 
By choosing $R_0$ to be arbitrary large, we can evaluate $\mathcal{G}(R)$ when $R\to \infty$. The first term is found as
\begin{align}
&\hspace{-0.4cm}\lim_{R\to \infty}\frac{1}{2}\Ex{f(w_{\parallel},w_{\perp})}{\log_2\left(\frac{\sigN}{\left(\sqrt{(R + w_{\parallel})^2 + w^2_{\perp}}+\mu\right)^2}+\sigd\right)}\notag\\&\hspace{5.2cm}=\frac{1}{2}\log_2\left(\sigd\right).\label{eq:func_g_zero}
\end{align}
The equality in~\eqref{eq:func_g_zero} follows from the dominated convergence theorem that permits interchange of the limit and the expectation operators.~\footnote{It is straightforward to show that the dominated convergence theorem holds by choosing  an integrable dominated (majorizer) function such as~$\log_2\left({\sigN}/{w^2_{\perp}}+\sigd\right)$ over the function inside the expectation.}

For the second term of $\mathcal{G}(R)$ we obtain
\begin{align}
&\lim_{R\to \infty} h\left(\sqrt{(R + w_{\parallel})^2 + w^2_{\perp}}\right)\notag\\
&=\lim_{R\to \infty} h\left(R + w_{\parallel}+O\left(\frac{1}{R}\right)\right)\label{eq:func_g_first_1}\\
&=\lim_{R\to \infty} h\left(R + w_{\parallel}\right)\label{eq:func_g_first_2}\\
&=\frac{1}{2}\log_2 2\pi e \sigN \label{eq:func_g_first_3}.
\end{align}
In \eqref{eq:func_g_first_1}, we have written the power series of the function inside the differential entropy about $R=\infty$. The $O(1/R)$ term represents omitted terms of order $1/R$. In \eqref{eq:func_g_first_2}, we used~\cite[Lemma.~6.9]{lapidoth03-10a}
\begin{align}
	&\lim_{\epsilon\to 0} h(a+\epsilon b)=h(a).\label{eq:entropy_plus_epsilon}
\end{align}
Finally, in \eqref{eq:func_g_first_3}, we used that translation does not change  differential entropy. Then, we used the entropy of a Gaussian distributed random variable~\cite[p.~244]{cover06-a}.
%
%

Similarly, for the third term of~$\mathcal{G}(R)$, we obtain
\begin{align}
\lim_{R\to \infty}h\left(\arctan \frac{w_{\perp}}{R+w_{\parallel}} +\Delta \big | r \right)&=h\left(\Delta\right)\\
&=\frac{1}{2}\log_2 2\pi e \sigd.\label{eq:func_g_second}
\end{align}
Here in \eqref{eq:func_g_second}, we again used \eqref{eq:entropy_plus_epsilon}, and the fact that 
\begin{align}
\lim_{R\to \infty} \arctan \frac{w_{\perp}}{R+w_{\parallel}}=0. 
\end{align}

Substituting \eqref{eq:func_g_zero}, \eqref{eq:func_g_first_3}, and \eqref{eq:func_g_second} in~\eqref{eq:upper_bound_ecape_0}, we obtain
\begin{align}
C^{R_0}(\SNR) \leq \min_{\mu\geq 0}\Big\{\frac{\betaU(\mu)}{\ln(2)}(\Es +2\sigN)& \notag\\&\hspace{-1.2cm}- \frac{1}{2}\sigN e^2\log_2\alphaU^2(\mu)\Big\}.
\label{upper_bound_escape_min}
\end{align}
From~\eqref{eq:pdf upper} we observe that~$q(r)$ becomes asymptotically independent of~$\mu$ when~$R\to \infty$, and thus become~$\alphaU(\mu)$ and $\betaU(\mu)$. Hence in this case,~$\alphaU(\mu)$ and~$\betaU(\mu)$ can be found for any arbitrary~$\mu$, and the minimization in~\eqref{upper_bound_escape_min} can be omitted. Without loss of generality, we set $\mu=0$, and substitute~\eqref{upper_bound_escape_min} in \eqref{eq:escape_to_infinity}
\begin{align}
&C(\SNR)\leq \frac{\betaU(\mu=0)}{\ln(2)}(\Es +2\sigN) \notag \\
&\hspace{2.5cm}- \frac{1}{2}\log_2\sigN e^2\alphaU^2(\mu=0)+ o(1).
\end{align}
This concluded the proof of Proposition~\ref{prop:upper_bound}.

%
\section{Swapping Supremum and Limit Operations}
\label{sec:sup_lim_swap}
In this section, we describe the steps involved in swapping the supremum and the limit operations in~\eqref{eq:inf lower_bound_4}. We first define 
\begin{align}\label{eq:swap_def1}
G(f(\Rt),n) \triangleq  \frac{1}{M}h(\rt^{(n/M)})-\frac{1}{n} h \left(\{a_k\}_{k=1}^{n}|\br\right).
\end{align}
For each fixed~$f^*(\Rt)$, $\underset{f(\Rt)}{\sup} G(f(\Rt),n) \geq G(f^*(\Rt),n)$. If~$\lim_{n \to \infty}G(f(\Rt),n)$ exists, we obtain  
\begin{align}\label{eq:swap_lim_1}
\lim_{n \to \infty} \underset{f(\Rt)}{\sup} G(f(\Rt),n)\geq \lim_{n \to \infty} G(f^*(\Rt),n).
\end{align}
Since~\eqref{eq:swap_lim_1} holds for any~$f^*(\Rt)$, the supremum also satisfies the inequality, thus
\begin{align}\label{eq:swap_final}
\lim_{n \to \infty} \underset{f(\Rt)}{\sup} G(f(\Rt),n)\geq \underset{f(\Rt)}{\sup} \lim_{n \to \infty} G(f(\Rt),n).
\end{align}
\section{Optimized Input Distribution}
\label{sec:functional_optimization}
\subsection{Background on Functional Optimization}
\label{sec:background_functional_optimization}
Definition of the capacity stated in \eqref{eq:capacity def} is a functional optimization problem~\cite{byron1992mathematics}. A functional is a mapping from a function to a real number, e.g. an integral. A background on \emph{calculus of variations} \cite{byron1992mathematics} is presented here that later will be used in our analysis.

Consider the functional $F(u)$, defined as
\begin{align}\label{eq:functional_K}
F(u) = \int_\Omega K(\bx,u(\bx)) \, \text{d}\bx,
\end{align}
where $u(\bx)$ is a real valued function of a real vector argument $\bx$,
\begin{align}
u: \Omega \subset \mathbb{R}^N \rightarrow \mathbb{R}.
\end{align}
A necessary condition for $u$ to be a stationary point of $F(u)$ under the $m$ constraints
\begin{align}\label{eq:functional_L}
\int_\Omega L_i(\bx,u(\bx)) \, \text{d}\bx = 0 , \: i=1,2,...,m \, ,
\end{align}
is that the following simplified Euler-Lagrange equation is satisfied
\begin{align}
\frac{\partial K}{\partial u} + \sum \limits_{i=1}^m \lambda_i \frac{\partial L_i}{\partial u}  = 0. \label{eq:EL}
\end{align}
The Lagrange multipliers $\lambda_i$ should be chosen to fulfill the constraints. 
Note that, $K$ and $L_i$, $i=1,2,...,m$, are real valued functions with continuous first partial derivatives.
%
\subsection{Functional Optimization of the Lower Bound}
We perform the supremum by employing the functional optimization method reviewed in Appendix~\ref{sec:background_functional_optimization}. Based on \eqref{eq:inf lower_bound_6}, the functional that must be maximized is written as
\begin{align}
F(f(\rt))=\int_0^\infty \Big[&-\frac{1}{M}f(\rt)\log(f(\rt)) \notag\\&-f(\rt) {\frac{1}{2}\log_2 2\pi e~g_\br(\rt)}\notag\\&+ \log_2 2\pi-\frac{1}{2}\log_2 2\pi e \sigN\Big] \text{d}\rt. \label{eq:capacity_upper_bound_R_functional}
\end{align}

Using \eqref{eq:capacity_upper_bound_R_functional}, the function corresponding to $K$ in \eqref{eq:functional_K} is identified as
\begin{align}\label{eq:K_u_lowerBound}
K(\rt,f(\rt)) &= -\frac{1}{M}f(\rt)\log_2 (f(\rt){(g_\br(\rt))^{M/2}}) -\frac{1}{2}\log_2 \sigN e^2.
\end{align}
There are two constraints on $f(\rt)$ that must be satisfied. First,~$f(\rt)$ has to integrate to one. The second constraint can be formulated based on the average power constraint of the input distribution. Consequently, for $L_i$ in \eqref{eq:functional_L}, we obtain 
\begin{IEEEeqnarray}{rCL}
\label{eq:L_U_lowerBound}
L_1(\rt,f(\rt)) &=& f(\rt) - 1 \label{eq:L_U1_lowerBound}\IEEEyesnumber\IEEEyessubnumber\\
L_2(\rt,f(\rt)) &=& ||\rt||^2 f(\rt) - M(\Es + \sigN).\label{eq:L_U2_lowerBound}\IEEEyessubnumber
\end{IEEEeqnarray}
Finally, by substituting \eqref{eq:K_u_lowerBound} and \eqref{eq:L_U_lowerBound} into the Euler-Lagrange equation \eqref{eq:EL}, we obtain the output distribution that maximizes \eqref{eq:inf lower_bound_6} as 
\begin{align}
f(\rt) &= \alphaL^{(M)} (g_\br(\rt))^{M/2} e^{-(\betaL^{(M)} ||\rt||^2)}. \label{eq:pdf lower}
\end{align}

As mentioned before, the optimized~$f(\rt)$ is used as~$f(\Rt)$ to evaluate the lower bound for the capacity of the channel in~\eqref{eq:channel}.
\end{appendices}

\section*{ACKNOWLEDGMENT}
Discussions with Ashkan Panahi are gratefully
acknowledged.
\bibliographystyle{IEEEtran}
\bibliography{IEEEabrv,references}

\begin{thebibliography}{10}
\providecommand{\url}[1]{#1}
\csname url@samestyle\endcsname
\providecommand{\newblock}{\relax}
\providecommand{\bibinfo}[2]{#2}
\providecommand{\BIBentrySTDinterwordspacing}{\spaceskip=0pt\relax}
\providecommand{\BIBentryALTinterwordstretchfactor}{4}
\providecommand{\BIBentryALTinterwordspacing}{\spaceskip=\fontdimen2\font plus
\BIBentryALTinterwordstretchfactor\fontdimen3\font minus
  \fontdimen4\font\relax}
\providecommand{\BIBforeignlanguage}[2]{{%
\expandafter\ifx\csname l@#1\endcsname\relax
\typeout{** WARNING: IEEEtran.bst: No hyphenation pattern has been}%
\typeout{** loaded for the language `#1'. Using the pattern for}%
\typeout{** the default language instead.}%
\else
\language=\csname l@#1\endcsname
\fi
#2}}
\providecommand{\BIBdecl}{\relax}
\BIBdecl

\bibitem{Viterbi1963}
A.~Viterbi, ``Phase-locked loop dynamics in the presence of noise by
  fokker-planck techniques,'' \emph{{\normalfont in}~Proc. IEEE}, vol.~51,
  no.~12, pp. 1737--1753, Dec. 1963.

\bibitem{dallal92}
Y.~E. Dallal and S.~Shamai, ``Performance bounds for noncoherent detection
  under {Brownian} phase noise,'' \emph{{IEEE} Trans. Inf. Theory}, vol.~38,
  no.~2, pp. 362--379, Mar. 1992.

\bibitem{Tomba98}
L.~Tomba, ``On the effect of {Wiener} phase noise in {OFDM} systems,''
  \emph{IEEE Trans. Commun.}, vol.~46, no.~5, pp. 580--583, May 1998.

\bibitem{Colavolpe2005}
G.~Colavolpe, A.~Barbieri, and G.~Caire, ``Algorithms for iterative decoding in
  the presence of strong phase noise,'' \emph{IEEE J. Sel. Areas Commun.},
  vol.~23, pp. 1748--1757, Sep. 2005.

\bibitem{syrjala2009phase}
V.~Syrj\"al\"a, M.~Valkama, N.~N. Tchamov, and J.~Rinne, ``Phase noise
  modelling and mitigation techniques in {OFDM} communications systems,''
  \emph{{\normalfont in}~Proc. Wireless Telecommun. Symp., (WTS)}, pp. 1--7,
  2009.

\bibitem{syrjala2011ofdm}
V.~Syrj{\"a}l{\"a}, M.~Valkama, Y.~Zou, N.~N. Tchamov, and J.~Rinne, ``On ofdm
  link performance under receiver phase noise with arbitrary spectral shape,''
  \emph{{\normalfont in}~Proc. IEEE Wireless Commun. and Netw. Conf. (WCNC)},
  pp. 1948--1953, Mar. 2011.

\bibitem{Khanzadi2011}
\text{M.R. Khanzadi}, H.~Mehrpouyan, E.~Alpman, T.~Svensson, D.~Kuylenstierna,
  and T.~Eriksson, ``On models, bounds, and estimation algorithms for
  time-varying phase noise,'' \emph{{\normalfont in}~Proc. Int. Conf. Signal
  Process. Commun. Syst. (ICSPCS)}, pp. 1--8, Dec. 2011.

\bibitem{Mehrpouyan2012}
H.~Mehrpouyan, A.~A. Nasir, S.~D. Blostein, T.~Eriksson, G.~K. Karagiannidis,
  and T.~Svensson, ``Joint estimation of channel and oscillator phase noise in
  {MIMO} systems,'' \emph{IEEE Trans. Signal Process.}, vol.~60, no.~9, pp.
  4790--4807, Sep. 2012.

\bibitem{Krishnan2012_1}
R.~Krishnan, \text{M.R. Khanzadi}, L.~Svensson, T.~Eriksson, and T.~Svensson,
  ``Variational bayesian framework for receiver design in the presence of phase
  noise in {MIMO} systems,'' \emph{{\normalfont in}~Proc. IEEE Wireless Commun.
  and Netw. Conf. (WCNC)}, pp. 1--6, Apr. 2012.

\bibitem{Khanzadi2013_0_COMP}
\text{M.R. Khanzadi}, R.~Krishnan, and T.~Eriksson, ``Effect of synchronizing
  coordinated base stations on phase noise estimation,'' \emph{{\normalfont
  in}~Proc. IEEE Acoust., Speech, Signal Process. (ICASSP)}, pp. 4938--4942,
  May 2013.

\bibitem{Khanzadi2013_1}
\text{M.R. Khanzadi}, D.~Kuylenstierna, A.~Panahi, T.~Eriksson, and H.~Zirath,
  ``Calculation of the performance of communication systems from measured
  oscillator phase noise,'' \emph{IEEE Trans. Circuits Syst. I, Reg. Papers},
  vol.~61, no.~5, pp. 1553--1565, May 2014.

\bibitem{syrjala2014analysis}
V.~Syrj\"al\"a, M.~Valkama, L.~Anttila, T.~Riihonen, and D.~Korpi, ``Analysis
  of oscillator phase-noise effects on self-interference cancellation in
  full-duplex {OFDM} radio transceivers,'' \emph{IEEE Trans. Wireless Commun.},
  vol.~13, no.~6, pp. 2977--2990, Jun. 2014.

\bibitem{Bjornson_MAMIMO2014_1}
E.~Bj{\"o}rnson, J.~Hoydis, M.~Kountouris, and M.~Debbah, ``Massive {MIMO}
  systems with non-ideal hardware: Energy efficiency, estimation, and capacity
  limits,'' \emph{{IEEE} Trans. Inf. Theory}, vol.~60, no.~11, pp. 7112--7139,
  Nov. 2014.

\bibitem{Krishnan14-01a}
\BIBentryALTinterwordspacing
R.~Krishnan, \text{M.R. Khanzadi}, N.~Krishnan \emph{et~al.}, ``On the impact
  of oscillator phase noise on the uplink performance in a massive {MIMO-OFDM}
  system,'' \emph{Submitted to IEEE Signal Process. Lett.}, May 2014. [Online].
  Available: \url{http://arxiv.org/abs/1405.0669}
\BIBentrySTDinterwordspacing

\bibitem{krishnan2015_VT_linear_MIMO}
------, ``Linear massive {MIMO} precoders in the presence of phase noise - {A}
  large-scale analysis,'' \emph{IEEE Trans. Veh. Technol.}, 2015, to appear.

\bibitem{Krishnan2013_1}
R.~Krishnan, \text{M.R. Khanzadi}, T.~Eriksson, and T.~Svensson, ``Soft metrics
  and their performance analysis for optimal data detection in the presence of
  strong oscillator phase noise,'' \emph{IEEE Trans. Commun.}, vol.~61, no.~6,
  pp. 2385--2395, Jun. 2013.

\bibitem{Krishnan13_2}
R.~Krishnan, A.~Graell~i Amat, T.~Eriksson, and G.~Colavolpe, ``Constellation
  optimization in the presence of strong phase noise,'' \emph{{IEEE} Trans.
  Commun.}, vol.~61, no.~12, pp. 5056--5066, Dec. 2013.

\bibitem{Smulders2002GHz}
P.~Smulders, ``Exploiting the {60 GHz} band for local wireless multimedia
  access: prospects and future directions,'' \emph{IEEE Commun. Mag.}, vol.~40,
  no.~1, pp. 140--147, Jan. 2002.

\bibitem{li2003high}
Y.~Li, H.~Jacobsson, M.~Bao, and T.~Lewin, ``High-frequency {SiGe MMICs}-an
  industrial perspective,'' \emph{{\normalfont in}~Proc. GigaHertz 2003 Symp.,
  \emph{Link\"{o}ping, Sweden}}, pp. 1--4, Nov. 2003.

\bibitem{Dohler2011}
M.~Dohler, R.~Heath, A.~Lozano, C.~Papadias, and R.~Valenzuela, ``Is the {PHY}
  layer dead?'' \emph{IEEE Commun. Mag.}, vol.~49, no.~4, pp. 159--165, Apr.
  2011.

\bibitem{Mehrpouyan2014_EBAND}
H.~Mehrpouyan, M.~Khanzadi, M.~Matthaiou, A.~Sayeed, R.~Schober, and Y.~Hua,
  ``Improving bandwidth efficiency in {E}-band communication systems,''
  \emph{IEEE Commun. Mag.}, vol.~52, no.~3, pp. 121--128, Mar. 2014.

\bibitem{khanzadi2014_highFreq}
M.~R. Khanzadi, R.~Krishnan, D.~Kuylenstierna, and T.~Eriksson, ``Oscillator
  phase noise and small-scale channel fading in higher frequency bands,''
  \emph{{\normalfont in}~Proc. IEEE Global Commun. Conf. (GLOBECOM)}, pp.
  410--415, Dec. 2014.

\bibitem{lapidoth02-10a}
A.~Lapidoth, ``On phase noise channels at high {SNR},'' \emph{{\normalfont
  in}~Proc. Inf. Theory Workshop (ITW)}, pp. 1--4, Oct. 2002.

\bibitem{colavolpe2001capacity}
G.~Colavolpe and R.~Raheli, ``The capacity of the noncoherent channel,''
  \emph{Europ. Trans. Telecommun.}, vol.~12, no.~4, pp. 289--296, Jul./Aug.
  2001.

\bibitem{peleg98-05a}
M.~Peleg and S.~{Shamai (Shitz)}, ``On the capacity of the blockwise incoherent
  {MPSK} channel,'' \emph{{IEEE} Trans. Inf. Theory}, vol.~46, no.~5, pp.
  603--609, May 1998.

\bibitem{nuriyev05-03a}
R.~Nuriyev and A.~Anastasopoulos, ``Capacity and coding for the
  block-independent noncoherent {AWGN} channel,'' \emph{{IEEE} Trans. Inf.
  Theory}, vol.~51, no.~3, pp. 866--883, Mar. 2005.

\bibitem{katz04-10a}
M.~Katz and S.~Shamai~(Shitz), ``On the capacity-achieving distribution of the
  discrete-time noncoherent and partially coherent {AWGN} channels,''
  \emph{{IEEE} Trans. Inf. Theory}, vol.~50, no.~10, pp. 2257--2270, Oct. 2004.

\bibitem{durisi12-08a}
G.~Durisi, ``On the capacity of the block-memoryless phase-noise channel,''
  \emph{{IEEE} Commun. Lett.}, vol.~16, no.~8, pp. 1157--1160, Aug. 2012.

\bibitem{hou02-05a}
P.~Hou, B.~Belzer, and T.~Fischer, ``Shaping gain of the partially coherent
  additive white {Gaussian} noise channel,'' \emph{{IEEE} Commun. Lett.},
  vol.~6, no.~5, pp. 175--177, May 2002.

\bibitem{barletta11-11a}
L.~Barletta, M.~Magarini, and A.~Spalvieri, ``Estimate of information rates of
  discrete-time first-order {Markov} phase noise channels,'' \emph{{IEEE}
  Photon. Technol. Lett.}, vol.~23, no.~21, pp. 1582--1584, Nov. 2011.

\bibitem{barbieri11-12a}
A.~Barbieri and G.~Colavolpe, ``On the information rate and repeat-accumulate
  code design for phase noise channels,'' \emph{{IEEE} Trans. Commun.},
  vol.~59, no.~12, pp. 3223--3228, Dec. 2011.

\bibitem{barletta12-05a}
L.~Barletta, M.~Magarini, and A.~Spalvieri, ``The information rate transferred
  through the discrete-time wiener's phase noise channel,'' \emph{J. Lightw.
  Technol.}, vol.~30, no.~10, pp. 1480--1486, May 2012.

\bibitem{gokceoglu2013mutual}
A.~Gokceoglu, Y.~Zou, M.~Valkama, P.~C. Sofotasios, P.~Mathecken, and
  D.~Cabric, ``Mutual information analysis of ofdm radio link under phase
  noise, iq imbalance and frequency-selective fading channel,'' \emph{IEEE
  Trans. Wireless Commun.}, vol.~12, no.~6, pp. 3048--3059, Jun. 2013.

\bibitem{ghozlan13-07a}
H.~Ghozlan and G.~Kramer, ``On {Wiener} phase noise channels at high
  signal-to-noise ratio,'' \emph{{\normalfont in}~Proc. IEEE Int. Symp. Inf.
  Theory (ISIT)}, pp. 2279--2283, Jul. 2013.

\bibitem{ghozlan14-01a}
------, ``Phase modulation for discrete-time {Wiener} phase noise channels with
  oversampling at high {SNR},'' \emph{{\normalfont in}~Proc. IEEE Int. Symp.
  Inf. Theory (ISIT)}, pp. 1554--1557, Jun. 2014.

\bibitem{barletta14-01a}
L.~Barletta and G.~Kramer, ``On continuous-time white phase noise channels,''
  \emph{{\normalfont in}~Proc. IEEE Int. Symp. Inf. Theory (ISIT)}, pp.
  2426--2429, Jun. 2014.

\bibitem{durisi2013capacity}
G.~Durisi, A.~Tarable, C.~Camarda, and G.~Montorsi, ``On the capacity of {MIMO}
  {Wiener} phase-noise channels,'' \emph{{\normalfont in}~Proc. Inf. Theory
  Applicat. Workshop (ITA)}, pp. 1--7, Feb. 2013.

\bibitem{durisi13-09b}
G.~Durisi, A.~Tarable, C.~Camarda, R.~Devassy, and G.~Montorsi, ``Capacity
  bounds for {MIMO} microwave backhaul links affected by phase noise,''
  \emph{{IEEE} Trans. Commun.}, Sep. 2013, submitted for publication.

\bibitem{durisi2013multiplexing}
G.~Durisi, A.~Tarable, and T.~Koch, ``On the multiplexing gain of {MIMO}
  microwave backhaul links affected by phase noise,'' \emph{{\normalfont
  in}~Proc. IEEE Int. Conf. Commun.(ICC)}, pp. 3209--3214, Jun. 2013.

\bibitem{Khanzadi14-T01a}
\text{M.R. Khanzadi}, G.~Durisi, and T.~Eriksson, ``Capacity of {SIMO} and
  {MISO} phase-noise channels with common/separate oscillators,'' \emph{IEEE
  Trans. Commun.}, Feb. 2015, to appear.

\bibitem{Khanzadi2013_2_ColoredPNEst}
\text{M.R. Khanzadi}, R.~Krishnan, and T.~Eriksson, ``Estimation of phase noise
  in oscillators with colored noise sources,'' \emph{IEEE Commun. Lett.},
  vol.~17, no.~11, pp. 2160--2163, Nov. 2013.

\bibitem{Demir2000}
A.~Demir, A.~Mehrotra, and J.~Roychowdhury, ``Phase noise in oscillators: a
  unifying theory and numerical methods for characterization,'' \emph{IEEE
  Trans. Circuits Syst. I, Fundam. Theory Appl.}, vol.~47, no.~5, pp. 655--674,
  May 2000.

\bibitem{cover06-a}
T.~M. Cover and J.~A. Thomas, \emph{Elements of Information Theory},
  2nd~ed.\hskip 1em plus 0.5em minus 0.4em\relax New York, NY, U.S.A.: Wiley,
  2006.

\bibitem{moser09-06a}
S.~M. Moser, ``{The fading number of multiple-input multiple-output fading
  channels with memory},'' \emph{{IEEE} Trans. Inf. Theory}, vol.~55, no.~6,
  pp. 2716--2755, Jun. 2009.

\bibitem{lapidoth03-10a}
A.~Lapidoth and S.~M. Moser, ``Capacity bounds via duality with applications to
  multiple-antenna systems on flat-fading channels,'' \emph{{IEEE} Trans. Inf.
  Theory}, vol.~49, no.~10, pp. 2426--2467, Oct. 2003.

\bibitem{topsoe67-a}
F.~Tops{\o}e, ``An information theoretical identity and a problem involving
  capacity,'' \emph{Studia Scientiarum Math. Hung.}, vol.~2, pp. 291--292,
  1967.

\bibitem{papoulis2002probability}
A.~Papoulis and S.~U. Pillai, \emph{Probability, random variables, and
  stochastic processes}.\hskip 1em plus 0.5em minus 0.4em\relax New York:
  McGraw-Hill, 2002.

\bibitem{beirlant1997nonparametric}
J.~Beirlant, E.~J. Dudewicz, L.~Gy{\"o}rfi, and E.~C. Van~der Meulen,
  ``Nonparametric entropy estimation: An overview,'' \emph{Int. J. Math.
  Statist. Sci.}, vol.~6, no.~1, pp. 17--39, 1997.

\bibitem{arnold06-08a}
D.~Arnold, H.-A. Loeliger, P.~Vontobel, A.~Kavcic, and W.~Zeng,
  ``Simulation-based computation of information rates for channels with
  memory,'' \emph{{IEEE} Trans. Inf. Theory}, vol.~52, no.~8, pp. 3498--3508,
  Aug. 2006.

\bibitem{dauwels2008computation}
J.~Dauwels and H.-A. Loeliger, ``Computation of information rates by particle
  methods,'' \emph{{IEEE} Trans. Inf. Theory}, vol.~54, no.~1, pp. 406--409,
  2008.

\bibitem{robert2004monte}
C.~P. Robert and G.~Casella, \emph{Monte Carlo statistical methods},
  2nd~ed.\hskip 1em plus 0.5em minus 0.4em\relax New York, NY: Springer-Verlag,
  2004.

\bibitem{lapidoth06-02a}
A.~Lapidoth and S.~M. Moser, ``The fading number of single-input
  multiple-output fading channels with memory,'' \emph{{IEEE} Trans. Inf.
  Theory}, vol.~52, no.~2, pp. 437--453, Feb. 2006.

\bibitem{byron1992mathematics}
F.~Byron and R.~Fuller, \emph{Mathematics of Classical and Quantum
  Physics}.\hskip 1em plus 0.5em minus 0.4em\relax Dover Publications, 1992.

\end{thebibliography}
\begin{IEEEbiography}[{\includegraphics[width=1in,height=1.25in,clip,keepaspectratio]{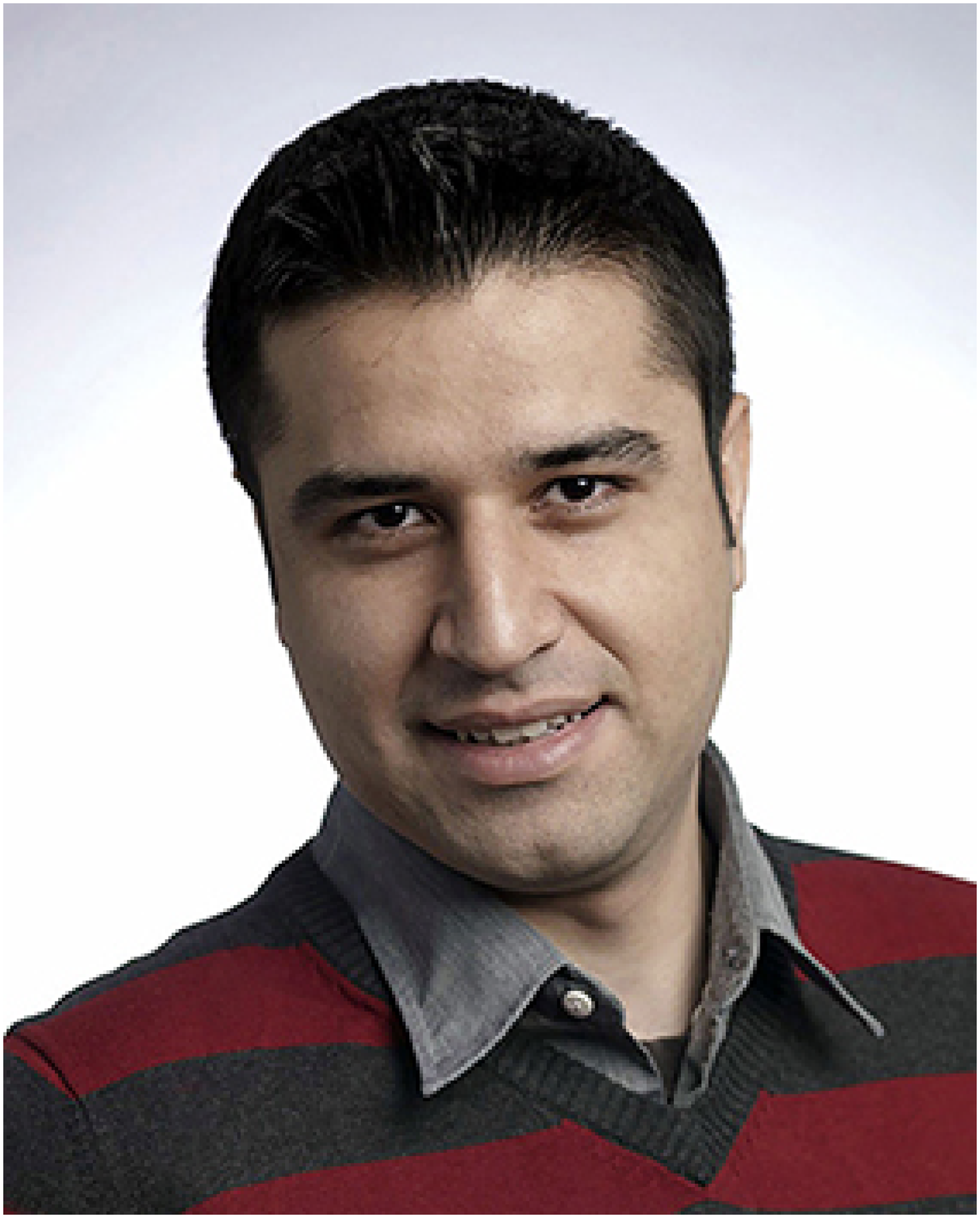}}]{M. Reza Khanzadi}(S'10)
received his M.Sc. degree in communication engineering from Chalmers University of Technology, Gothenburg, Sweden, in 2010. He is currently a Ph.D. candidate at the Department of Signals and Systems in collaboration with the Department of Microtechnology and Nanoscience of the same university. From October to December 2014, he was a Research Visitor in the University of Southern California, Los Angeles, CA. 
Bayesian inference, statistical signal processing, and information theory are his current research interests. His Ph.D. project is mainly focused on radio frequency oscillator modeling, oscillator phase noise estimation/compensation, and determining the effect of oscillator phase noise on the performance of communication systems. 
He has been the recipient of the S2 Pedagogical Prize 2012 from Department of Signals and Systems, as well as 2013, 2014 and 2015 Ericsson's Research Foundation grants, 2014 Chalmers Foundation grant, and 2014 Gothenburg Royal Society of Arts and Sciences grant. 
\end{IEEEbiography}

\begin{IEEEbiography}[{\includegraphics[width=1in,height=1.25in,clip,keepaspectratio]{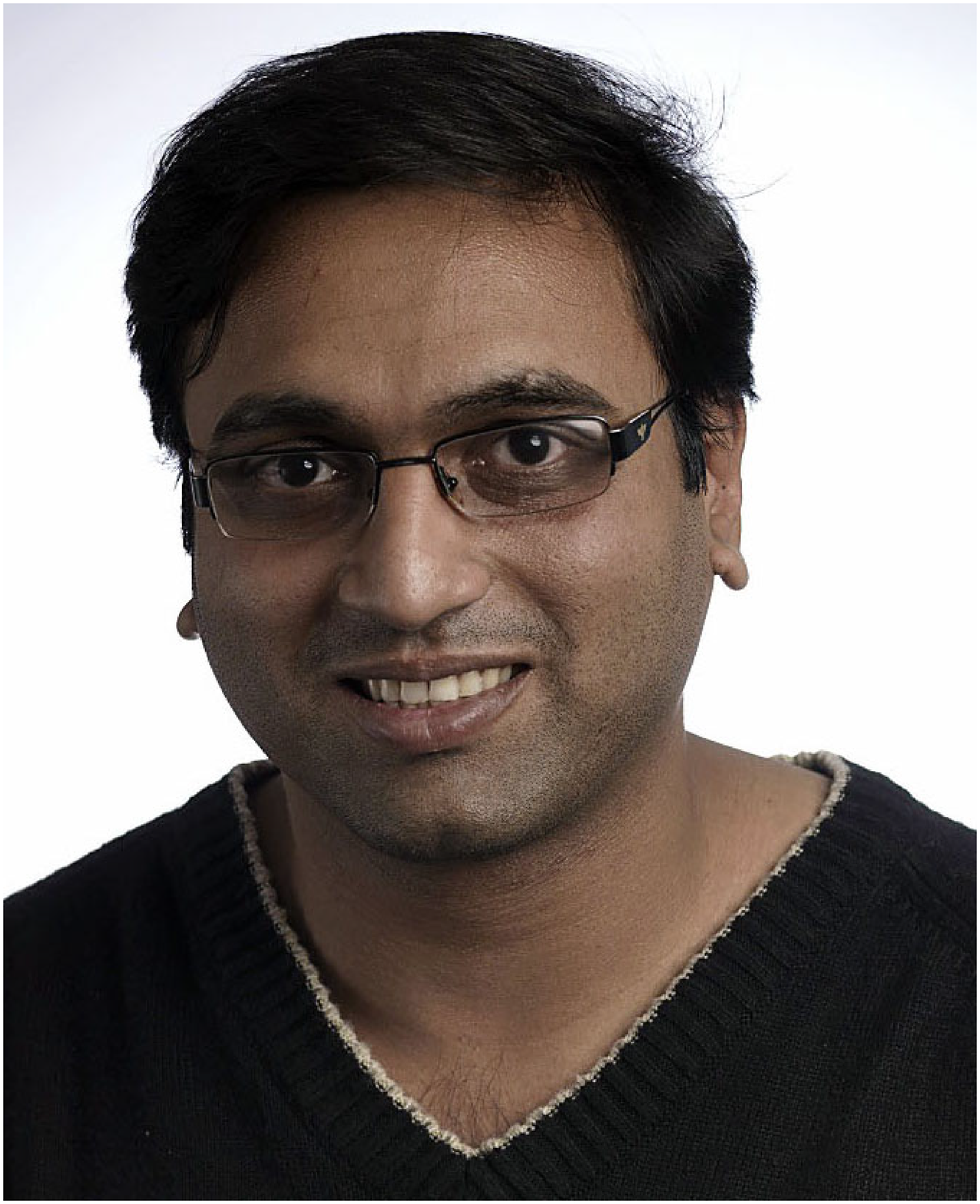}}]{Rajet Krishnan} completed his B.Sc. degree in electronics and communication engineering from College of Engineering Trivandrum India in 2004, his M.Sc. degree in electrical and computer engineering from Kansas State University, USA in 2009, and his Ph.D. in communication systems from Chalmers University of Technology, Sweden in 2015.  His research interests include communication theory, statistical signal processing, Bayesian Inference, information theory, and machine learning.
\end{IEEEbiography}

\begin{IEEEbiography}[{\includegraphics[width=1in,height=1.25in,clip,keepaspectratio]{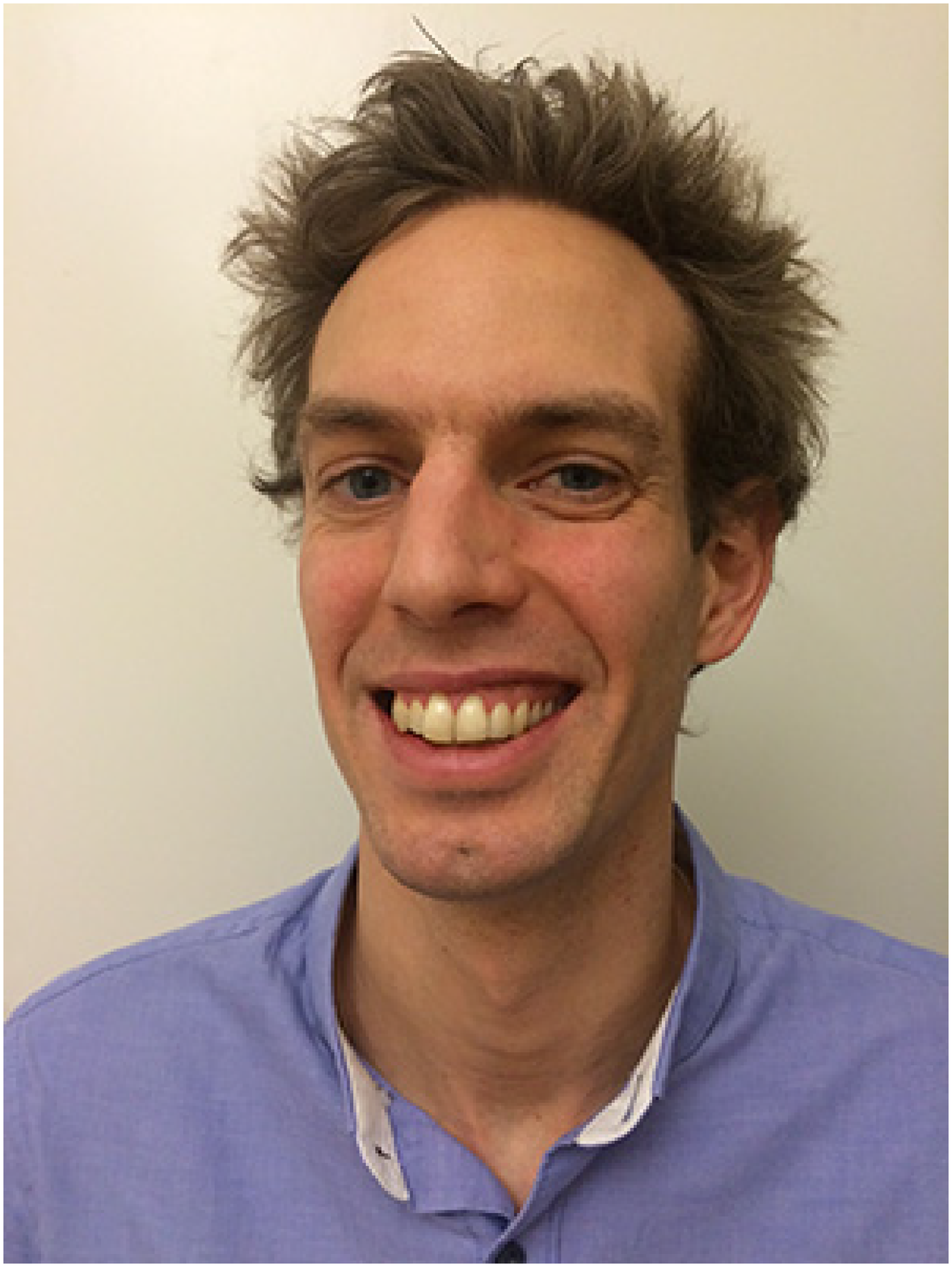}}]{Johan S\"{o}der} received the B.Sc. degree in engineering physics in 2007 and the M.Sc. degree in communication engineering in 2010, both from Chalmers University of Technology. He joined Ericsson in 2010 working with the LTE product development. Since 2012 he has worked within Ericsson Research focusing on performance evaluations and algorithm development for wireless access networks.
\end{IEEEbiography}

\begin{IEEEbiography}[{\includegraphics[width=1in,height=1.25in,clip,keepaspectratio]{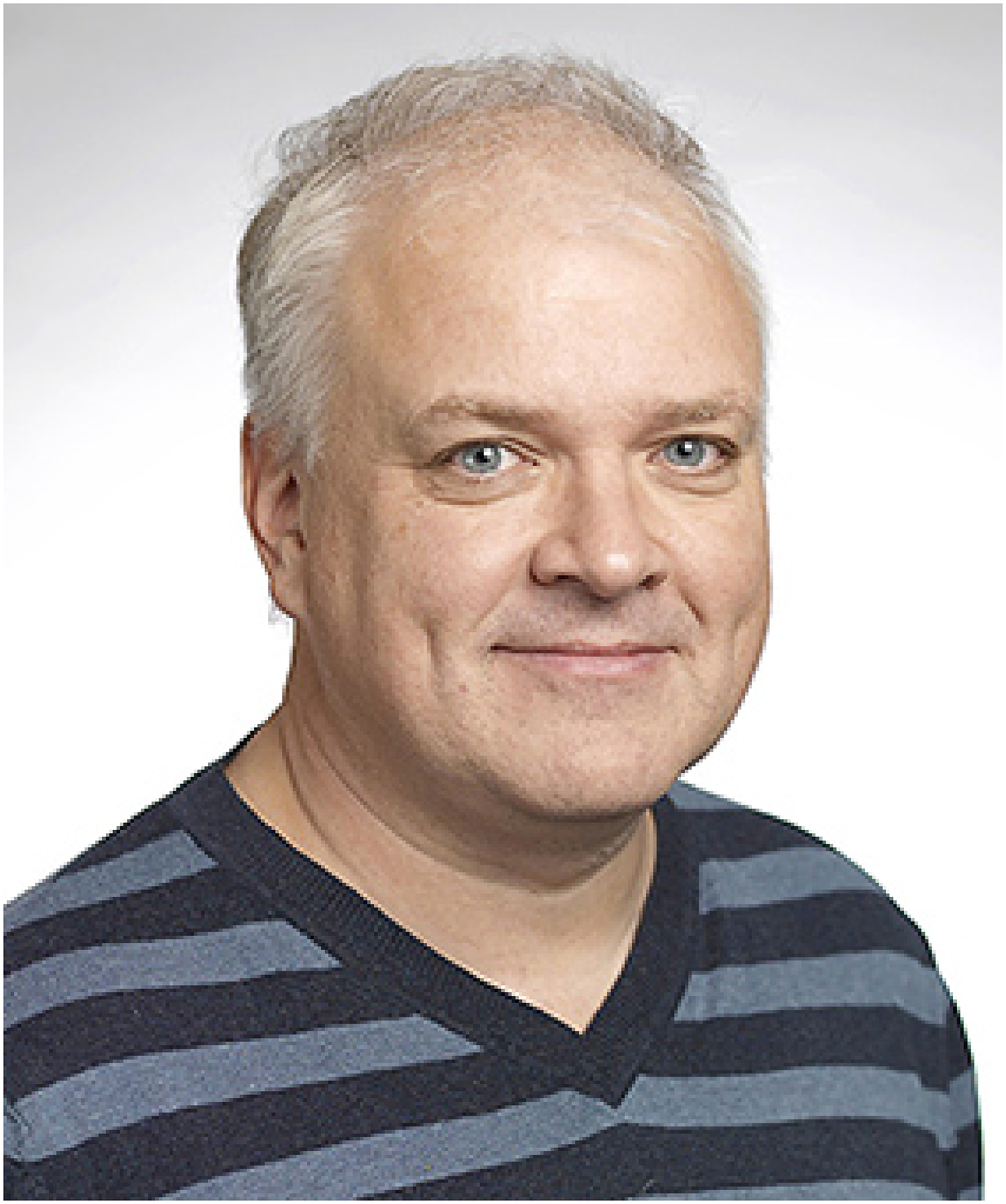}}]{Thomas Eriksson} received the Ph.D. degree in Information Theory in 1996, from Chalmers University of Technology, Gothenburg, Sweden. From 1990 to 1996, he was at Chalmers. In 1997 and 1998, he was at AT\&T Labs - Research in Murray Hill, NJ, USA, and in 1998 and 1999 he was at Ericsson Radio Systems AB, Kista, Sweden. Since 1999, he has been at Chalmers University, where he is a professor in communication systems. Further, he was a guest professor at Yonsei University, S. Korea, in 2003-2004. He is currently vice head of the department of Signals and Systems at Chalmers, with responsibility for undergraduate and master education. His research interests include communication, data compression, and modeling and compensation of non-ideal hardware components (e.g. amplifiers, oscillators, modulators in communication transmitters and receivers).
\end{IEEEbiography}
\end{document}